\pdfoutput=1
\documentclass[acmsmall,screen]{acmart}

\setcopyright{cc}
\setcctype{by}
\acmDOI{10.1145/3808255}
\acmYear{2026}
\acmJournal{PACMPL}
\acmVolume{10}
\acmNumber{PLDI}
\acmArticle{177}
\acmMonth{6}
\acmSubmissionID{pldi26main-p24-p}
\received{2025-11-04}
\received[accepted]{2026-04-03}



\usepackage{braket}
\usepackage{cleveref}
\usepackage{enumitem}
\usepackage{listings}
\usepackage{mathpartir}
\usepackage{mathtools}
\usepackage{pgfplots}
\usepackage{rotating}
\usepackage{siunitx}
\usepackage{subcaption}
\usepackage{stmaryrd}
\usepackage{textcase}
\usepackage{threeparttable}
\usepackage{tikz}
\usepackage{wrapfig}
\usetikzlibrary{positioning}
\usetikzlibrary{calc}
\usetikzlibrary{quantikz2}
\usetikzlibrary{patterns}
\usepgfplotslibrary{colorbrewer}
\pgfplotsset{compat=1.18}

\newcommand{\tikzsetfigurename}[1]{}

\makeatletter
\def\@secfont{\sffamily\bfseries\section@raggedright\MakeTextUppercase}
\makeatother

\lstdefinelanguage{cobble}
{morekeywords={
     Basic, Sum, Tensor, Poly, kron, subnormalization, total_cost, optimize, queries
   },
 sensitive=true,
 morecomment=[n]{/*}{*/},
 morecomment=[l]{\#},
 morestring=[b]", 
 escapechar=\%,
 columns=fullflexible,
 keepspaces=true,
 basicstyle=\ttfamily,
 mathescape=true,
}
\DeclarePairedDelimiter\abs{\lvert}{\rvert}%
\DeclarePairedDelimiter\norm{\lVert}{\rVert}

\crefname{lemma}{Lemma}{Lemmas}
\Crefname{lemma}{Lemma}{Lemmas}

\AddToHook{env/definition/begin}{\crefalias{theorem}{definition}}
\AddToHook{env/example/begin}{\crefalias{theorem}{example}}
\AddToHook{cmd/appendix/before}{%
    \crefalias{section}{appendix}%
    \crefalias{subsection}{appendix}
}
\creflabelformat{equation}{#2\textup{#1}#3}

\numberwithin{equation}{section}

\newcommand{\oracle}[1]{\textcolor{orange!80!black}{#1}}

\begin{document}

\author{Charles Yuan}
\affiliation{
  \institution{University of Wisconsin--Madison}
  \country{USA}
}
\email{charlesyuan@cs.wisc.edu}
\orcid{0000-0002-4918-4467}

\newcommand{\LangName}{Cobble}
\title{\LangName{}: Compiling Block Encodings for Quantum Computational Linear Algebra}

\begin{abstract}
Quantum algorithms for computational linear algebra promise up to exponential speedups for applications such as simulation and regression, making them prime candidates for hardware realization. But these algorithms execute in a model that cannot efficiently store matrices in memory like a classical algorithm does, instead requiring developers to implement complex expressions for matrix arithmetic in terms of correct and efficient quantum circuits. Among the challenges for the developer is navigating a cost model in which conventional optimizations for linear algebra, such as subexpression reuse, can be inapplicable or unprofitable.

In this work, we present \LangName{}, a language for programming with quantum computational linear algebra. \LangName{} enables developers to express and manipulate the quantum representations of matrices, known as block encodings, using high-level notation that automatically compiles to correct quantum circuits.
\LangName{} features analyses that compute the time and space usage of programs, as well as optimizations that reduce overhead and generate efficient circuits using state-of-the-art techniques such as the quantum singular value transformation.
We evaluate \LangName{} on benchmark kernels for simulation, regression, search, and other applications, showing $2.6\times$--$25.4\times$ speedups on these benchmarks compared to the unoptimized baseline.
\end{abstract}

\maketitle

\section{Introduction} \label{sec:introduction}

Linear algebra is among the most promising applications of a quantum computer.
A quantum state of $n$ qubits encodes a vector of $2^n$ elements, enabling quantum algorithms to solve certain problems in linear systems \citep{harrow2009}, physical simulation \citep{childs2012}, regression \citep{chakraborty2023}, and differential equations \citep{berry2014} in time polylogarithmic in the problem dimension. In suitable settings, the speedup over classical methods is exponential.

Practically realizing these speedups, however, poses a programming challenge even as hardware continues to advance \citep{bluvstein2023,google2025} toward the eventual availability of a scalable, fault-tolerant quantum computer.
A quantum algorithm cannot efficiently store arbitrary matrices in memory like a classical algorithm does, because just processing every element would itself take linear time, limiting the possible speedup. Instead, quantum algorithms manipulate implicit representations of matrices that exploit structure, such as sparsity and factorization, using a tool called \emph{block encoding} \citep{low2019,gilyen2019,martyn2021}.

\paragraph{Block Encoding}
Every quantum state is a unit vector, and every circuit of quantum logic gates is a unitary matrix.
A \emph{block encoding} of a matrix $A$ is a unitary matrix with $A$ embedded in its top left block.
This concept enables a quantum circuit to implement a non-unitary $A$ as a part of a larger unitary computation, to multiply $A$ with a vector by applying the circuit to a quantum state, and to perform arithmetic on matrices given that $A + B$ may not be unitary even if $A$ and $B$ are.

A research paper specifies a quantum algorithm $f$ on an input $A$ by giving matrix arithmetic to turn a block encoding of $A$ into $f(A)$. In a linear solver, $f(A) = A^{-1}$; in simulation \citep{gilyen2019}, $f(A) = e^{-iAt}$; and in search \citep{martyn2021}, $f(A) = \mathrm{sign}(A) = A (A^2){}^{-1/2}$.
A paper typically explains how $f(A)$ can be built from arithmetic operators, e.g.\ $A^{-1} = I + (I - A) + (I - A){}^2 + \cdots$.
Many papers \citep{sunderhauf2024,camps2024} also show how a block encoding of the desired input $A$ can be built from sums $B + C$, products $B \cdot C$, and tensor products $B \otimes C$ of realizable parts such as primitive logic gates or black-box circuits given as oracles.

\paragraph{Abstraction Challenge}
To translate the block encoding arithmetic in an algorithms paper into a program to run on hardware, a developer could search for circuit constructions for matrix operators such as addition.
A hurdle is that existing quantum programming languages such as Silq \citep{bichsel2020}, Q\# \citep{svore2018}, and Qiskit \citep{javadiabhari2024} do not provide abstractions for creating and manipulating block encodings. They instead enable the developer to explicitly specify bit-level logic gates for each matrix expression, which for many linear algebra applications would be millions of gates \citep{scherer2017}.
Block encodings present an opportunity to design new abstractions that ease this burden for developers of libraries and applications.

\paragraph{Efficiency Challenge}
Reasoning about block encodings is essential not only to specify but also to optimize quantum algorithms, where they often act as performance bottlenecks \citep{nibbi2024,li2023}.
This reasoning is hard due to an unconventional cost model. The cost of a block encoding is given by not just its number of logic gates, but also a multiplier called \emph{subnormalization} that dictates how many repetitions of the program are needed to produce the correct answer. This factor is not obvious from the literal length of a program and must be calculated separately.

Moreover, different ways to build the same matrix may be cheaper or costlier. A simple example is $A$ vs. $\frac{1}{2}(A + A)$, which uses twice as many gates. A more subtle concern is that classical techniques to optimize linear algebra, such as subexpression reuse, can be unprofitable in quantum programs. For example, there is no general mechanism to compute a block encoding of $(A + B) \cdot (A + B)$ via just one addition and one multiplication by reusing the intermediate value of $A + B$.

\paragraph{Programming with Block Encodings}
To bridge the abstraction gap, we present \LangName{}, a language to build and compose block encodings.
\LangName{} offers high-level mathematical operators for matrix arithmetic, including sum, product, tensor product, and choice.
In \LangName{}, a developer specifies block encodings for basic matrices using standard gates or application-specific circuits, and then composes them together using these mathematical operators.
Like the Eigen \citep{gunnebaud2010} linear algebra library but unlike prior quantum circuit languages, \LangName{}~enables a developer to implement an algorithm using publication-style math notation rather than logic gates alone.

The \LangName{} compiler automatically translates a high-level program to a circuit to run on hardware, while the \LangName{} type system guarantees that every well-typed program has a valid circuit. \LangName{} also provides a cost model, derived from theory research, that a developer can use to compute the dominant factors in the time and space usage of programs, including subnormalization.

\paragraph{Optimizations for Block Encodings}
To tackle the efficiency challenge, we present two optimizations for \LangName{} programs that reduce gate and subnormalization costs.
The first, \emph{sum fusion}, flattens nested linear combinations of matrices to remove intermediate overhead from subnormalization.
The second, \emph{polynomial fusion}, combines repeated sums and products into a compact form that has an efficient circuit given by the quantum singular value transformation \citep{gilyen2019}.
These optimizations exploit the algebraic structure of block encodings, and are complementary to techniques for end-to-end circuit synthesis from arbitrary matrices \citep{hantzko2024}.

We implement \LangName{} as an embedded language in Python. To establish benchmarks for the emerging domain of quantum linear algebra, we use \LangName{} to express a set of kernels for simulation, regression, and other applications. Our optimizations reduce their total runtime cost -- gate count times subnormalization -- by $2.6\times$--$25.4\times$ over the unoptimized program. \LangName{} takes little compile time and can yield more speedup than existing circuit optimizers on these programs.

\paragraph{Contributions}
In this work, we present the following contributions:
\begin{enumerate}[labelwidth=!, labelsep=0.5em, leftmargin=3em, itemindent=0pt, listparindent=\parindent, align=parleft]
  \item[Sec.~\ref{sec:language}:] \LangName{}, a quantum programming language of mathematical operators over block-encoded matrices, in which well-typed programs compile to valid quantum circuits;
  \item[Sec.~\ref{sec:cost-model}:] A cost model to calculate the time and space usage of \LangName{} programs, indicating when classical optimizations are inapplicable or unprofitable in the quantum setting;
  \item[Sec.~\ref{sec:optimizations}:] The sum and polynomial fusion optimizations for \LangName{}, which underpin an optimizing rewrite system that is sound, cost-nonincreasing, and strongly normalizing; and
  \item[Sec.~\ref{sec:evaluation}:] An evaluation on benchmarks for simulation, regression, and other applications that shows $2.6\times$--$25.4\times$ reductions in total runtime cost over the unoptimized baseline.
\end{enumerate}

\paragraph{Impact}
Our work enables developers to express core components for quantum linear algebra using high-level notation rather than qubit-level circuits and to soundly reduce their runtime costs. These ideas pave way to more accurate estimates of hardware requirements, more useful and robust benchmarks and compilers, and more near-term demonstrations of quantum advantage.

\section{Background} \label{sec:background}

This section briefly reviews mathematical concepts needed for this work.
A reader familiar with quantum computation but not with quantum algorithms for linear algebra is encouraged to read \Cref{sec:qcla}.
For more detail, see the texts of \citet{lin2022} and \citet{nielsen_chuang_2010}.

\subsection{Quantum Computation}

A \emph{qubit} exists in a \emph{superposition} or linear combination $\lambda_0 \ket{\texttt{0}} + \lambda_1 \ket{\texttt{1}}$ of two vectors $\ket{\texttt{0}} = [1, 0]{}^\top$ and $\ket{\texttt{1}} = [0, 1]{}^\top$, where $\lambda_0, \lambda_1 \in \mathbb{C}$ are \emph{amplitudes} satisfying $\abs{\lambda_0}^2 + \abs{\lambda_1}^2 = 1$. Examples of qubits include $\ket{\texttt{0}}$, $\ket{\texttt{1}}$, and the states $\frac{1}{\sqrt{2}}{(\ket{\texttt{0}} + e^{i\varphi} \ket{\texttt{1}})}$ where $\varphi \in [0, 2\pi)$ is known as a \emph{phase}. 

A \emph{quantum state} of $n$ qubits is a superposition over $n$-bit strings. For example, $\smash{\frac{1}{\sqrt{2}}}{(\ket{\texttt{00}}+\ket{\texttt{11}})}$ is a quantum state over two qubits.
Formally, multiple component states combine via the \emph{tensor product} $\otimes$ of vectors, such that the state $\ket{\texttt{01}}$ is defined as $\ket{\texttt{0}} \otimes \ket{\texttt{1}}$. We use the customary notations $\ket{\texttt{01}}$, $\ket{\texttt{0}, \texttt{1}}$, and $\ket{\texttt{0}}\ket{\texttt{1}}$ to denote $\ket{\texttt{0}} \otimes \ket{\texttt{1}}$, and the notation $\ket{\texttt{0}}^{\otimes n}$ to denote $n$ copies of $\ket{\texttt{0}}$.

\paragraph{Unitary Operators}
A \emph{quantum logic gate} manipulates the bit strings and their amplitudes within a quantum state without collapsing the state from superposition.
The semantics of a quantum gate is a unitary matrix $U$ --- a linear, norm-preserving, and invertible operator with $U^{-1} = U^{\dagger}$.

The quantum gates over a single qubit include the NOT gate $X = \begin{bsmallmatrix} 0 & 1 \\ 1 & 0 \end{bsmallmatrix}$, mapping $\ket{x} \mapsto \ket{\texttt{1} - x}$ for $x \in \{\texttt{0}, \texttt{1}\}$; the phase flip gate $Z = \smash{\begin{bsmallmatrix} 1 & 0 \\ 0 & -1 \end{bsmallmatrix}}$, mapping $\ket{x} \mapsto (-1){}^x\ket{x}$; the $\pi/4$ rotation gate $T$, mapping $\ket{x} \mapsto e^{ix\pi/4}\ket{x}$; and the Hadamard gate $H$, mapping $\ket{x} \mapsto \frac{1}{\sqrt{2}}{(\ket{\texttt{0}} + (-1){}^x \ket{\texttt{1}})}$.

The effect of a gate may be \emph{controlled} by one or more qubits. For example, the two-qubit CNOT gate maps $\ket{\texttt{0}, x} \mapsto \ket{\texttt{0}, x}$ and $\ket{\texttt{1}, x} \mapsto \ket{\texttt{1}, \textrm{NOT}\ x} = \ket{\texttt{1}, \texttt{1}-x}$. The three-qubit Toffoli gate is the quantum analogue of AND, mapping $\ket{\texttt{1}, \texttt{1}, x} \mapsto \ket{\texttt{1}, \texttt{1}, \texttt{1} - x}$ only if the first two qubits are \texttt{1}.

\paragraph{Measurement}
A \emph{measurement} probabilistically collapses the superposition of a quantum state into a classical outcome.
When a qubit $\lambda_0 \ket{\texttt{0}} + \lambda_1 \ket{\texttt{1}}$ is measured in the standard basis, the observed classical outcome is $\texttt{0}$ with probability $\abs{\lambda_0}^2$ and $\texttt{1}$ with probability $\abs{\lambda_1}^2$.

A composite state is \emph{entangled} when it cannot be written as a tensor product of its components.
The \emph{Bell state} $\smash{\frac{1}{\sqrt{2}}}(\ket{\texttt{00}} + \ket{\texttt{11}})$ is entangled, since it cannot be written as a product of two separate qubits.
Given an entangled state, measuring one of its components causes the superposition of the other component to also collapse.
For example, measuring the second qubit in the Bell state causes the first to also collapse, to either $\ket{\texttt{0}}$ or $\ket{\texttt{1}}$ with probability $\big\lvert\frac{1}{\sqrt{2}}\big\rvert{}^2 = \frac{1}{2}$ each.

\subsection{Quantum Computational Linear Algebra} \label{sec:qcla}

The essence of quantum algorithms for linear algebra is the ability of a quantum state or unitary operator to encode an exponentially large vector or matrix respectively:

\begin{definition} \label{def:amplitude-encoding}
  Given a vector $x \in \mathbb{R}^N$ with $n = \log N$, its \emph{amplitude encoding} is the $n$-qubit state
  \[
    \ket{x} = \frac{1}{\norm{x}_2} \sum_{j=0}^{N-1} x_j \ket{j}.
  \]
\end{definition}

For example, two numbers can be encoded into the amplitudes of one qubit via the $R_y(\theta)$ gate, which takes $\ket{\texttt{0}} \mapsto \cos(\theta/2)\ket{\texttt{0}} + \sin(\theta/2)\ket{\texttt{1}}$.
Efficient circuits to precisely encode vectors with higher dimension are an active area of research. General techniques applicable to our work include rotation trees with lookup tables \citep{low2024} and alias sampling \citep{babbush2018}.

\begin{definition} \label{def:block-encoding}
  Given a matrix $A \in \mathbb{R}^{N \times N}$, a \emph{block encoding} of $A$ is any $(m+n)$-qubit unitary
  \[
    \mathcal{B}[A] =
    \begin{bmatrix}
      A / \alpha & \cdot \\
      \cdot & \cdot
    \end{bmatrix},
  \]
  where the top left block of the matrix is $A$ rescaled by a \emph{subnormalization} $\alpha$ satisfying $\alpha \ge \norm{A}_2$, the spectral norm of $A$. The remaining blocks can be arbitrary as long as $\mathcal{B}[A]$ is unitary.
  Intuitively, we embed a possibly non-unitary $A$ into any unitary $\mathcal{B}[A]$ representable by a quantum circuit.
\end{definition}

\begin{example}
  A unitary matrix $U$ has itself as a block encoding $\mathcal{B}[U] = U$, using $m = 0$ additional qubits and $\alpha = 1$.
  In general, non-unitary matrices require $m \ge 1$ qubits to encode. For example,
  \[
    \mathcal{B}[A] =
    \begin{bmatrix}
      A & \sqrt{I - A^2} \\
      \sqrt{I - A^2} & -A
    \end{bmatrix}
  \]
  is unitary for Hermitian $A$ with $\norm{A}_2 \le 1$, using $m = 1$ qubit. More qubits are needed in practice.
\end{example}

A unitary operator that block-encodes a matrix $A$ acts on a state that amplitude-encodes a vector $x$ by matrix-vector multiplication $Ax$.
Applying $\mathcal{B}[A]$ to $\ket{x}$ alongside $m$ copies of $\ket{\texttt{0}}$ yields
\begin{align} \label{eq:block-encoding}
\mathcal{B}[A] \left(\ket{\texttt{0}}\!^{\otimes m} \ket{x}\right) = \frac{\norm{Ax}_2}{\alpha}\ket{\texttt{0}}\!^{\otimes m} \ket{A x} + \ket{\bot},
\end{align}
where $\ket{Ax}$ is a normalized amplitude encoding of $Ax$ not scaled by $\alpha$. An algorithm such as that of \citet{harrow2009} can read out from this vector desired information such as its inner products.

\paragraph{Subnormalization}
The state $\ket{\bot}$ is an undesirable failure case. Extracting $\ket{Ax}$ from the superposition requires \emph{post-selection}: measuring the $m$ temporary qubits, accepting if all yield $\ket{\texttt{0}}$, and starting over otherwise.
Using amplitude amplification techniques \citep{berry2014b}, the expected number of rounds until success is $O(\alpha)$. Running time is thus proportional to subnormalization.

\paragraph{Matrix Arithmetic}
Given $\mathcal{B}[A]$ and $\mathcal{B}[B]$, circuits are known that construct $\mathcal{B}[A + B]$, $\mathcal{B}[A \cdot B]$, $\mathcal{B}[A \otimes B]$, and other arithmetic operators \citep{gilyen2019}. For exposition, we describe next the construction of \citet{childs2012} to block-encode linear combinations of matrices.

\begin{definition} \label{def:lcu}
  Given $(m+n)$-qubit operators that block-encode $A_j \in \mathbb{R}^{N \times N}$ and the coefficients $\lambda \in \mathbb{R}^L$ where $\ell = \log L$, a block encoding of $\sum_{j=0}^{L-1} \lambda_j A_j$ is the $(\ell+m+n)$-qubit operator
  \begin{align}
    \mathcal{B}\left[\sum_{j=0}^{L-1} \lambda_j A_j\right] &= \big(\textsc{Prepare}^\dagger \otimes I^{\otimes (m+n)}\big) \cdot \textsc{Select} \cdot \big(\textsc{Prepare} \otimes I^{\otimes (m+n)}\big), \label{eq:lcu} \\[-7pt]
    \text{where }\,\textsc{Prepare} \left(\ket{\texttt{0}}^{\otimes \ell}\right) &= \frac{1}{\sqrt{\norm{\lambda}_1}} \sum_{j=0}^{L-1} \sqrt{\abs{\lambda_j}} \ket{j},\, \text{and} \label{eq:prepare} \\[2pt]
    \textsc{Select} \left(\ket{j} \ket{\texttt{0}}^{\otimes m} \ket{x}\right) &= \mathrm{sign}(\lambda_j) \ket{j} \mathcal{B}[A_j] \left(\ket{\texttt{0}}^{\otimes m} \ket{x}\right). \label{eq:select}
  \end{align}
\end{definition}

Reading \Cref{eq:lcu} from right to left, the \textsc{Prepare} operator first creates an amplitude encoding (\Cref{def:amplitude-encoding}) into $\ket{j}$ of the coefficients $\smash{\left[\sqrt{\lambda_0}, \ldots, \sqrt{\lambda_{L-1}}\right]}$. Next, the \textsc{Select} operator chooses one of the $\mathcal{B}[A_j]$ to execute by controlling on the bits of $\ket{j}$. Finally, the inverse of the \textsc{Prepare} operator restores $\ket{j}$ to zero to enable post-selection. The reason for taking square roots of the coefficients is that \textsc{Prepare} and $\textsc{Prepare}^\dagger$ each incur a $\smash{\sqrt{\lambda_j}}$ factor, which multiply to give the desired $\lambda_j$. Negative coefficients $\lambda_j$ are handled by inserting a $Z$ phase flip gate on each corresponding branch.

This method is known as \emph{linear combination of unitaries} (LCU), so named for its original use case. It incurs a subnormalization that scales total runtime by $\alpha = \norm{\lambda}_1 = \sum_j \abs{\lambda_j}$. Negative $\lambda_j$ increase $\alpha$ because the denominator of \Cref{eq:prepare} sums positive squared norms of amplitudes $\smash{\sqrt{\abs{\lambda_j}}}$.

\paragraph{Quantum Singular Value Transformation}
A more general and efficient way to compute polynomials of block-encoded matrices is the \emph{quantum singular value transformation} (QSVT) of \citet{gilyen2019}.
Given a block encoding of a Hermitian matrix $A$ and a degree-$d$ real polynomial $P(x)$ with fixed (even or odd) parity satisfying $\abs{P(x)} \le 1$ for all $x \in [-1, 1]$, QSVT computes the polynomial $P(A)$ via a circuit that applies a sequence of rotation gates interleaved with instances of $\mathcal{B}[A]$:
\begin{equation} \label{eq:qsvt}
\mathcal{B}[P(A)] = e^{i \phi_0 Z_\Pi} \mathcal{B}[A] e^{i \phi_1 Z_\Pi} \cdots e^{i \phi_{d-1} Z_\Pi} \mathcal{B}[A] e^{i \phi_d Z_\Pi},
\end{equation}
in which the $Z_{\Pi}$ limits each phase rotation to the subspace $\ket{\texttt{0}}\!^{\otimes m}$ where all temporary qubits are zero. The phase angles $\phi_j$ are computed from the coefficients of $P(x)$ via a framework known as \emph{quantum signal processing} (QSP) developed by \citet{low2019,martyn2021}.

For conceptual simplicity, this work uses a special case of QSVT for Hermitian matrices $A$, also known as the \emph{quantum eigenvalue transformation}.
Matrix functions such as exponentials are more difficult to define for matrices with non-square dimensions and less well-behaved for non-Hermitian matrices with non-real or nonexistent eigenvalues.
\LangName{} is restricted to Hermitian matrices for these existence guarantees to apply; non-Hermitian cases would require different treatment (e.g.\ dilation to a larger Hermitian operator) and are beyond the immediate scope of this work.

\section{Example}

To illustrate programming in \LangName{} and reasoning about the costs of programs, in this section we walk through how to express and optimize block-encoded matrices for quantum applications. These examples are loosely derived from the algorithms literature. They are specifically chosen to demonstrate the system and preview the more complex applications we evaluate in \Cref{sec:evaluation}.

\subsection{Simulation and Sum Fusion Optimization} \label{sec:simulation-example}
Consider the simulation of a system of particles, such as atoms or photons, that permits two distinct operations on pairs of particles.
The first operation swaps the energy of two particles, while the second excites or suppresses both at once \citep{roth2017}.
In the form of a \emph{Hamiltonian}, the total energy function of the system, the two operations could be expressed as the $4 \times 4$ matrices
\begin{equation*}
A = X \otimes X + Y \otimes Y, \quad B = X \otimes X - Y \otimes Y,
\end{equation*}
where $X = \begin{bsmallmatrix} 0 & 1 \\ 1 & 0 \end{bsmallmatrix}$ and $Y = \smash{\begin{bsmallmatrix} 0 & -i \\ i & 0 \end{bsmallmatrix}}$ coincide with basic single-qubit logic gates.
Then, a system that performs $A$ and $B$ simultaneously with relative intensities $\lambda_A$ and $\lambda_B \in \mathbb{R}$ is a weighted sum
\begin{equation} \label{eq:two-particle}
C = \lambda_A A + \lambda_B B.
\end{equation}

To find how this system evolves over time, a quantum simulation algorithm \citep{low2019} multiplies an initial state vector by a function, such as an exponential, of the matrix $C$. It can then use measurements on the final state to calculate energy or other quantities of interest.

\paragraph{Direct Implementation}
To concretely implement this algorithm as a program, a developer must build a block encoding of the matrix $C$.
The developer can express $C$ in \LangName{} directly according to the mathematical notation above. For parameters $\lambda_A = 1$ and $\lambda_B = 0.3$, the program is:
\begin{lstlisting}
X, Y = Basic("X"), Basic("Y")
A = kron(X, X) + kron(Y, Y)
B = kron(X, X) - kron(Y, Y)
C = A + 0.3 * B
\end{lstlisting}

\begin{figure}
\tikzsetfigurename{two-particle}
\scalebox{0.73}{\begin{quantikz}
\lstick{$\ket{\texttt{0}}$} & \gate{R_y(2\cos^{-1} \sqrt{1/1.3})} & \octrl{1} & \octrl{3} & \octrl{3} & \octrl{1} & \ctrl{1} & \ctrl{1} & \ctrl{3} & \ctrl{3} & \ctrl{1} & \gate{R_y(-2\cos^{-1} \sqrt{1/1.3})} & \ground{} \rstick{$\ket{\texttt{0}}$} \\
\lstick{$\ket{\texttt{0}}$} & \qw & \gate{H}\gategroup[3, steps=4, style={dashed,gray,rounded corners,inner xsep=2pt}, label style={label position=south east,anchor=south east,yshift=-0.22cm}]{$\textcolor{gray}{\mathcal{B}[A]}$} & \octrl{2} & \ctrl{2} & \gate{H} & \gate{H}\gategroup[3, steps=5, style={dashed,gray,rounded corners,inner xsep=2pt}, label style={label position=south east,anchor=south east,yshift=-0.22cm}]{$\textcolor{gray}{\mathcal{B}[B]}$} & \gate{Z} & \octrl{2} & \ctrl{2} & \gate{H} & \qw & \ground{} \rstick{$\ket{\texttt{0}}$} \\
\lstick[2]{$\ket{x}$} & \qw & \qw & \gate[style={fill=orange!20}]{X} & \gate[style={fill=orange!20}]{Y} & \qw & \qw & \qw & \gate[style={fill=orange!20}]{X} & \gate[style={fill=orange!20}]{Y} & \qw & \qw & \rstick[2]{$\ket{Cx}$} \\
& \qw & \qw & \gate[style={fill=orange!20}]{X} & \gate[style={fill=orange!20}]{Y} & \qw & \qw & \qw & \gate[style={fill=orange!20}]{X} & \gate[style={fill=orange!20}]{Y} & \qw & \qw & \qw
\end{quantikz}}

\caption{Quantum circuit to directly implement the system in \Cref{eq:two-particle}, as produced by \LangName{}. The gates in \oracle{orange} are queries to $\mathcal{B}[X] = \oracle{X}$ and $\mathcal{B}[Y] = \oracle{Y}$. Filled $\bullet$ denotes control on $\ket{\texttt{1}}$ while hollow $\circ$ denotes control on $\ket{\texttt{0}}$. Notation $\smash{\begin{quantikz}\qw & \ground{}\end{quantikz}}$ denotes post-selection: measuring  a qubit and starting over until $\ket{\texttt{0}}$ is observed.} \label{fig:two-particle}

\vspace*{1.25ex}%
\hspace*{-1.15em}
\begin{minipage}{0.39\textwidth}
\begin{lstlisting}[language=Python,basicstyle=\fontsize{5}{5}\selectfont\ttfamily,numbers=none,morekeywords={MCMTGate,XGate,YGate,HGate,QuantumCircuit,ry,ch,cz}]
c = QuantumCircuit(4)
c.ry(2 * arccos(sqrt(1 / 1.3)), 3)
c.append(HGate().control(1, ctrl_state=0), [3, 2])
c.append(MCMTGate(XGate(),2,2,ctrl_state='00'),[3,2,1,0])
c.append(MCMTGate(YGate(),2,2,ctrl_state='10'),[3,2,1,0])
c.append(HGate().control(1, ctrl_state=0), [3, 2])
c.ch(3, 2)
c.cz(3, 2)
c.append(MCMTGate(XGate(),2,2,ctrl_state='01'),[3,2,1,0])
c.append(MCMTGate(YGate(),2,2,ctrl_state='11'),[3,2,1,0])
c.ch(3, 2)
c.ry(-2 * arccos(sqrt(1 / 1.3)), 3)
\end{lstlisting}
\end{minipage}\hspace{0.3em}
\begin{minipage}{0.62\textwidth}
\tikzsetfigurename{two-particle-opt}
\scalebox{0.73}{\begin{quantikz}
\lstick{$\ket{\texttt{0}}$} & \gate{R_y(2\cos^{-1} \sqrt{1.3/2})} & \ctrl[open]{2} & \ctrl{2} & \gate{R_y(-2\cos^{-1} \sqrt{1.3/2})} & \ground{} \rstick{$\ket{\texttt{0}}$} \\
\lstick[2]{$\ket{x}$} & \qw & \gate[style={fill=orange!20}]{X} & \gate[style={fill=orange!20}]{Y} & \qw & \rstick[2]{$\ket{Cx}$} \\
& \qw & \gate[style={fill=orange!20}]{X} & \gate[style={fill=orange!20}]{Y} & \qw & \qw
\end{quantikz}}
\end{minipage}

\setlength{\abovecaptionskip}{5pt}
\hspace*{-0.75em}
\begin{minipage}{0.38\textwidth}
\caption{Qiskit code for \Cref{fig:two-particle}.} \label{fig:two-particle-qiskit}
\end{minipage}
\begin{minipage}{0.60\textwidth}
\caption{Optimized circuit for \Cref{eq:two-particle} after sum fusion.} \label{fig:two-particle-opt}
\end{minipage}
\end{figure}

The \LangName{} compiler translates this program into the circuit in \Cref{fig:two-particle}, whose matrix form is $\mathcal{B}[C]$ with scaled $C$ in the top-left block (\Cref{def:block-encoding}).
Following the LCU method (\Cref{def:lcu}), the circuit first prepares an amplitude encoding of the vector $\smash{\left[\sqrt{1}, \sqrt{0.3}\right]}$ in an \emph{ancilla}, or temporary, qubit.
It uses this ancilla to control sub-circuits for $\mathcal{B}[A]$ and $\mathcal{B}[B]$, derived recursively.
Finally, it reverses the preparation of the ancilla so that post-selection (see \Cref{sec:qcla}) realizes $C$.

Using \LangName{}, the developer need not explicitly list the bit-level rotation and controlled gates in \Cref{fig:two-particle} as they would when constructing a quantum circuit in a language such as Qiskit \citep{javadiabhari2024}, depicted for contrast in \Cref{fig:two-particle-qiskit}. They can instead use mathematical notation to express an application that compiles, behind the scenes, to a circuit to execute on hardware.

\paragraph{Cost of Direct Implementation}

The time complexity of the program is proportional to the number of logic gates in the circuit. Because the precise gate count is implementation-dependent, we may instead count the number of \emph{queries} to the block encodings of the basic matrices that appear in \Cref{eq:two-particle}. Queries are highlighted in \Cref{fig:two-particle}, where they are basic $\oracle{X}$ and $\oracle{Y}$ logic gates.

The time complexity is also proportional to the \emph{subnormalization} $\alpha$ (see \Cref{sec:qcla}) that scales the encoded $C$ and gives the expected number of repetitions of the circuit until its post-selection succeeds. As stated in \Cref{def:lcu}, for the linear combinations $A$ and $B$ we have $\alpha_A = \abs{1} + \abs{1} = 2$ and $\alpha_B = \abs{1} + \abs{-1} = 2$. For $C$, the subnormalization accumulates as $\alpha_C = \alpha_A + 0.3 \cdot \alpha_B = 2.6$.

Overall, to successfully compute the block encoding of $C$, the program must execute 8 queries to $\oracle{X}$ and $\oracle{Y}$ per iteration of the circuit and at least 2.6 iterations in expectation, for a total cost of 20.8 queries.
Using the \LangName{} system, the developer can compute these costs automatically:

\begin{lstlisting}[numbers=none]
>>> C.queries(), C.subnormalization(), C.total_cost()
(8.0, 2.6, 20.8)
\end{lstlisting}

\paragraph{Sum Fusion Optimization}

The key idea of the \emph{sum fusion} optimization is to flatten the nesting of linear combinations containing negative coefficients that cancel out. In the example,
\begin{align*}
C &= 1 \cdot (X \otimes X + Y \otimes Y) + 0.3 \cdot (X \otimes X - Y \otimes Y) \\
&= 1.3 \cdot \oracle{X} \otimes \oracle{X} + 0.7 \cdot \oracle{Y} \otimes \oracle{Y},
\end{align*}
which is equivalent but invokes fewer queries to $\oracle{X}$ and $\oracle{Y}$ and has a lower subnormalization.

Given the original program \texttt{C}, \LangName{} can automatically perform sum fusion and a set of related rewrites to produce a new program with lower cost.
\Cref{fig:two-particle-opt} presents the circuit that the \LangName{} compiler generates after optimization.
The optimized program makes only 4 queries to the basic matrices $\oracle{X}$ and $\oracle{Y}$.
It also has a smaller subnormalization $\abs{1.3} + \abs{0.7} = 2.0$, for an overall reduction of $2.6\times$ in total cost.
Once again, the \LangName{} system derives this information automatically:

\begin{lstlisting}[numbers=none]
>>> C2 = C.optimize(); (C2.queries(), C2.subnormalization(), C2.total_cost())
(4.0, 2.0, 8.0)
\end{lstlisting}

\subsection{Regression and Polynomial Fusion Optimization} \label{sec:regression-example}

Consider a regression analysis of data measured by a quantum sensor, as could be done in quantum-enhanced learning of physical systems \citep{huang2022}.
Denoting the dataset as $A$ and the model as $B$, regression minimizes their error $A - B$\@.
Suppose that for regularization, a quantum regression algorithm \citep{chakraborty2023} adapts Huber-like loss functions \citep{huber1964} that interpolate between linear and squared error and balance terms that exaggerate or dampen the error:
\begin{align}
  f &= (A - B) + \textstyle\frac{1}{2} {(A - B)}^2, \quad g = (A - B) - \textstyle\frac{1}{2} {(A - B)}^2, \nonumber \\
  L &= f \cdot g. \label{eq:loss}
\end{align}

The algorithm computes a block encoding of $L$ to find how well the model describes the sensor data.
It accesses the dataset via a physical process denoted by a black-box unitary operator $\oracle{U_A} = \mathcal{B}[A]$, and the model via a black-box sub-circuit denoted by a unitary operator $\oracle{U_B} = \mathcal{B}[B]$.

\paragraph{Direct Implementation}
Given these black-box inputs, a developer can express $L$ in \LangName{} following the mathematical notation, and automatically obtain the LCU-based circuit in \Cref{fig:qsp-circuit}:
\begin{lstlisting}
A, B = Basic("U_A", num_ancillas=1), Basic("U_B", num_ancillas=1)
f = (A - B) + 1 / 2 * (A - B) ** 2
g = (A - B) - 1 / 2 * (A - B) ** 2
L = f * g
\end{lstlisting}

\begin{figure}
\tikzsetfigurename{qsp-circuit}
\resizebox{\textwidth}{!}{
\begin{quantikz}
\lstick{$\ket{\texttt{0}}$} & \qw & \qw & \qw & \qw & \qw & \qw & \qw & \qw & \qw & \qw & \qw & \qw & \qw & \qw & \qw & \qw & \qw & \qw & \qw & \qw & \targ{} & \qw & \qw & \qw & \qw & \qw & \qw & \qw & \qw & \qw & \qw & \qw & \qw & \qw & \qw & \qw & \qw & \qw & \qw & \qw & \gate{X} & \ground{} \rstick{$\ket{\texttt{0}}$} \\
\lstick{$\ket{\texttt{0}}$} & \gate{H} & \gate{Z} & \octrl{2} & \octrl{2} & \octrl{3} & \octrl{3} & \octrl{2} & \ctrl{2} & \ctrl{2} & \ctrl{3} & \ctrl{3} & \ctrl{2} & \ctrl{1} & \ctrl{2} & \ctrl{2} & \ctrl{3} & \ctrl{3} & \ctrl{2} & \ctrl{1} & \gate{H} & \octrl{-1} & \gate{H} & \octrl{2} & \octrl{2} & \octrl{3} & \octrl{3} & \octrl{2} & \ctrl{2} & \ctrl{2} & \ctrl{3} & \ctrl{3} & \ctrl{2} & \ctrl{1} & \ctrl{2} & \ctrl{2} & \ctrl{3} & \ctrl{3} & \ctrl{2} & \ctrl{1} & \gate{H} & \qw & \ground{} \rstick{$\ket{\texttt{0}}$} \\
\lstick{$\ket{\texttt{0}}$} & \qw & \qw & \qw & \qw & \qw & \qw & \qw & \qw & \qw & \qw & \qw & \qw & \targ{} & \qw & \qw & \qw & \qw & \qw & \targ{} & \qw & \octrl{-2} & \qw & \qw & \qw & \qw & \qw & \qw & \qw & \qw & \qw & \qw & \qw & \targ{} & \qw & \qw & \qw & \qw & \qw & \targ{} & \qw & \qw & \ground{} \rstick{$\ket{\texttt{0}}$} \\
\lstick{$\ket{\texttt{0}}$} & \qw & \qw & \gate{H} & \gate{Z} & \octrl{1} & \ctrl{1} & \gate{H} & \gate{H} & \gate{Z} & \octrl{1} & \ctrl{1} & \gate{H} & \octrl{-1} & \gate{H} & \gate{Z} & \octrl{1} & \ctrl{1} & \gate{H} & \qw & \qw & \octrl{-3} & \qw & \gate{H} & \gate{Z} & \octrl{1} & \ctrl{1} & \gate{H} & \gate{H} & \gate{Z} & \octrl{1} & \ctrl{1} & \gate{H} & \octrl{-1} & \gate{H} & \gate{Z} & \octrl{1} & \ctrl{1} & \gate{H} & \qw & \qw & \qw & \ground{} \rstick{$\ket{\texttt{0}}$} \\
\lstick{$\ket{\texttt{0}}$} & \qw & \qw & \qw & \qw & \gate[wires=2,style={fill=orange!20}]{U_A} & \gate[wires=2,style={fill=orange!20}]{U_B} & \qw & \qw & \qw & \gate[wires=2,style={fill=orange!20}]{U_A} & \gate[wires=2,style={fill=orange!20}]{U_B} & \qw & \octrl{-2} & \qw & \qw & \gate[wires=2,style={fill=orange!20}]{U_A} & \gate[wires=2,style={fill=orange!20}]{U_B} & \qw & \qw & \qw & \octrl{-4} & \qw & \qw & \qw & \gate[wires=2,style={fill=orange!20}]{U_A} & \gate[wires=2,style={fill=orange!20}]{U_B} & \qw & \qw & \qw & \gate[wires=2,style={fill=orange!20}]{U_A} & \gate[wires=2,style={fill=orange!20}]{U_B} & \qw & \octrl{-2} & \qw & \qw & \gate[wires=2,style={fill=orange!20}]{U_A} & \gate[wires=2,style={fill=orange!20}]{U_B} & \qw & \qw & \qw & \qw & \ground{} \rstick{$\ket{\texttt{0}}$} \\
\lstick{$\ket{x}$} & \qw & \qw & \qw & \qw &  &  & \qw & \qw & \qw &  &  & \qw & \qw & \qw & \qw &  &  & \qw & \qw & \qw & \qw & \qw & \qw & \qw &  &  & \qw & \qw & \qw &  &  & \qw & \qw & \qw & \qw &  &  & \qw & \qw & \qw & \qw & \rstick{$\ket{Lx}$}
\end{quantikz}}
\caption{Quantum circuit to directly implement the function in \Cref{eq:loss}, as produced by \LangName{}. The gates in \oracle{orange} are queries to the black-box operators $\mathcal{B}[A] = \oracle{U_A}$ and $\mathcal{B}[B] = \oracle{U_B}$, assumed to use one ancilla.} \label{fig:qsp-circuit}
\end{figure}

\paragraph{Sum Fusion}
Next, the developer can rewrite the program using sum fusion. Distributing the product and collecting like terms results in a new expression, written compactly as:
\begin{align*}
  L &= \left((A - B) + \textstyle\frac{1}{2} {(A - B)}^2\right) \left((A - B) - \textstyle\frac{1}{2} {(A - B)}^2\right) = {(\oracle{A} - \oracle{B})}^2 - \textstyle\frac{1}{4}{(\oracle{A} - \oracle{B})}^4.
\end{align*}

\Cref{fig:qsp-sum-opt} depicts the circuit to which this expression compiles. Notably, sum fusion alone yields little benefit --- the new circuit is as complex as \Cref{fig:qsp-circuit}, despite the new $L$ being more concise.
The reason is that for general $C$ and $D$, the circuit to block-encode $C \cdot D$ sequences the circuits that encode $C$ and $D$\@.
\Cref{fig:qsp-sum-opt} sequences $\oracle{A} - \oracle{B}$ with itself $2 + 4 = 6$ times for 12 queries in total.

Unlike classical matrix expressions that can be evaluated to a numerical value and cached for reuse, there is no structure-independent mechanism to evaluate a sub-circuit in a quantum block encoding into a form that can be reused cheaply later.
Each instance of a matrix must be physically represented by a unitary operator encoding that matrix, whose cost is the same everywhere.

\begin{figure}
\tikzsetfigurename{qsp-sum-opt}
\resizebox{\textwidth}{!}{
\begin{quantikz}
\lstick{$\ket{\texttt{0}}$} & \gate{H} & \gate{Z} & \octrl{3} & \octrl{3} & \octrl{4} & \octrl{4} & \octrl{3} & \octrl{2} & \octrl{3} & \octrl{3} & \octrl{4} & \octrl{4} & \octrl{3} & \octrl{2} & \ctrl{3} & \ctrl{3} & \ctrl{4} & \ctrl{4} & \ctrl{3} & \ctrl{1} & \ctrl{2} & \ctrl{3} & \ctrl{3} & \ctrl{4} & \ctrl{4} & \ctrl{3} & \ctrl{1} & \ctrl{2} & \ctrl{3} & \ctrl{3} & \ctrl{4} & \ctrl{4} & \ctrl{3} & \ctrl{1} & \ctrl{2} & \ctrl{3} & \ctrl{3} & \ctrl{4} & \ctrl{4} & \ctrl{3} & \ctrl{2} & \ctrl{1} & \gate{H} & \ground{} \rstick{$\ket{\texttt{0}}$} \\
\lstick{$\ket{\texttt{0}}$} & \qw & \qw & \qw & \qw & \qw & \qw & \qw & \qw & \qw & \qw & \qw & \qw & \qw & \qw & \qw & \qw & \qw & \qw & \qw & \targ{} & \qw & \qw & \qw & \qw & \qw & \qw & \targ{} & \qw & \qw & \qw & \qw & \qw & \qw & \targ{} & \qw & \qw & \qw & \qw & \qw & \qw & \qw & \targ{} & \qw & \ground{} \rstick{$\ket{\texttt{0}}$} \\
\lstick{$\ket{\texttt{0}}$} & \qw & \qw & \qw & \qw & \qw & \qw & \qw & \targ{} & \qw & \qw & \qw & \qw & \qw & \targ{} & \qw & \qw & \qw & \qw & \qw & \ctrl{-1} & \targ{} & \qw & \qw & \qw & \qw & \qw & \ctrl{-1} & \targ{} & \qw & \qw & \qw & \qw & \qw & \ctrl{-1} & \targ{} & \qw & \qw & \qw & \qw & \qw & \targ{} & \qw & \qw & \ground{} \rstick{$\ket{\texttt{0}}$} \\
\lstick{$\ket{\texttt{0}}$} & \qw & \qw & \gate{H} & \gate{Z} & \octrl{1} & \ctrl{1} & \gate{H} & \octrl{-1} & \gate{H} & \gate{Z} & \octrl{1} & \ctrl{1} & \gate{H} & \qw & \gate{H} & \gate{Z} & \octrl{1} & \ctrl{1} & \gate{H} & \octrl{-2} & \octrl{-1} & \gate{H} & \gate{Z} & \octrl{1} & \ctrl{1} & \gate{H} & \octrl{-2} & \octrl{-1} & \gate{H} & \gate{Z} & \octrl{1} & \ctrl{1} & \gate{H} & \octrl{-2} & \octrl{-1} & \gate{H} & \gate{Z} & \octrl{1} & \ctrl{1} & \gate{H} & \qw & \qw & \qw & \ground{} \rstick{$\ket{\texttt{0}}$} \\
\lstick{$\ket{\texttt{0}}$} & \qw & \qw & \qw & \qw & \gate[wires=2,style={fill=orange!20}]{U_A} & \gate[wires=2,style={fill=orange!20}]{U_B} & \qw & \octrl{-2} & \qw & \qw & \gate[wires=2,style={fill=orange!20}]{U_A} & \gate[wires=2,style={fill=orange!20}]{U_B} & \qw & \qw & \qw & \qw & \gate[wires=2,style={fill=orange!20}]{U_A} & \gate[wires=2,style={fill=orange!20}]{U_B} & \qw & \octrl{-3} & \octrl{-2} & \qw & \qw & \gate[wires=2,style={fill=orange!20}]{U_A} & \gate[wires=2,style={fill=orange!20}]{U_B} & \qw & \octrl{-3} & \octrl{-2} & \qw & \qw & \gate[wires=2,style={fill=orange!20}]{U_A} & \gate[wires=2,style={fill=orange!20}]{U_B} & \qw & \octrl{-3} & \octrl{-2} & \qw & \qw & \gate[wires=2,style={fill=orange!20}]{U_A} & \gate[wires=2,style={fill=orange!20}]{U_B} & \qw & \qw & \qw & \qw & \ground{} \rstick{$\ket{\texttt{0}}$} \\
\lstick{$\ket{x}$} & \qw & \qw & \qw & \qw &  &  & \qw & \qw & \qw & \qw &  &  & \qw & \qw & \qw & \qw &  &  & \qw & \qw & \qw & \qw & \qw &  &  & \qw & \qw & \qw & \qw & \qw &  &  & \qw & \qw & \qw & \qw & \qw &  &  & \qw & \qw & \qw & \qw & \rstick{$\ket{Lx}$}
\end{quantikz}}
\caption{Intermediate circuit that \LangName{} produces for \Cref{eq:loss} after sum fusion.} \label{fig:qsp-sum-opt}
\end{figure}

\begin{figure}
\tikzsetfigurename{qsp-opt}
\resizebox{\textwidth}{!}{
\begin{quantikz}
\lstick{$\ket{\texttt{0}}$} & \gate{H} & \targ{} & \gate{R_z(-5.50)} & \targ{} & \qw & \qw & \qw & \qw & \qw & \targ{} & \gate{R_z(\pi)} & \targ{} & \qw & \qw & \qw & \qw & \qw & \targ{} & \gate{R_z(1.57)} & \targ{} & \qw & \qw & \qw & \qw & \qw & \targ{} & \gate{R_z(\pi)} & \targ{} & \qw & \qw & \qw & \qw & \qw & \targ{} & \gate{R_z(0.79)} & \targ{} & \gate{H} & \ground{} \rstick{$\ket{\texttt{0}}$} \\
\lstick{$\ket{\texttt{0}}$} & \qw & \octrl{-1} & \qw & \octrl{-1} & \gate{H} & \gate{Z} & \octrl{1} & \ctrl{1} & \gate{H} & \octrl{-1} & \qw & \octrl{-1} & \gate{H} & \gate{Z} & \octrl{1} & \ctrl{1} & \gate{H} & \octrl{-1} & \qw & \octrl{-1} & \gate{H} & \gate{Z} & \octrl{1} & \ctrl{1} & \gate{H} & \octrl{-1} & \qw & \octrl{-1} & \gate{H} & \gate{Z} & \octrl{1} & \ctrl{1} & \gate{H} & \octrl{-1} & \qw & \octrl{-1} & \qw & \ground{} \rstick{$\ket{\texttt{0}}$} \\
\lstick{$\ket{\texttt{0}}$} & \qw & \octrl{-2} & \qw & \octrl{-2} & \qw & \qw & \gate[wires=2,style={fill=orange!20}]{U_A} & \gate[wires=2,style={fill=orange!20}]{U_B} & \qw & \octrl{-2} & \qw & \octrl{-2} & \qw & \qw & \gate[wires=2,style={fill=orange!20}]{U_A} & \gate[wires=2,style={fill=orange!20}]{U_B} & \qw & \octrl{-2} & \qw & \octrl{-2} & \qw & \qw & \gate[wires=2,style={fill=orange!20}]{U_A} & \gate[wires=2,style={fill=orange!20}]{U_B} & \qw & \octrl{-2} & \qw & \octrl{-2} & \qw & \qw & \gate[wires=2,style={fill=orange!20}]{U_A} & \gate[wires=2,style={fill=orange!20}]{U_B} & \qw & \octrl{-2} & \qw & \octrl{-2} & \qw & \ground{} \rstick{$\ket{\texttt{0}}$} \\
\lstick{$\ket{x}$} & \qw & \qw & \qw & \qw & \qw & \qw &  &  & \qw & \qw & \qw & \qw & \qw & \qw &  &  & \qw & \qw & \qw & \qw & \qw & \qw &  &  & \qw & \qw & \qw & \qw & \qw & \qw &  &  & \qw & \qw & \qw & \qw & \qw & \rstick{$\ket{x}$}
\end{quantikz}}
\caption{Optimized circuit that \LangName{} produces for \Cref{eq:loss} after sum and polynomial fusion.} \label{fig:qsp-opt}
\end{figure}

\paragraph{Polynomial Fusion}
To make progress, \LangName{} exploits the structure of the repeated terms in a matrix polynomial.
First, it invokes a \emph{polynomial fusion} transformation that merges monomials with the same base and collects polynomial coefficients across the program.
Given the program \texttt{L}, \LangName{} uses the symbolic term $\textsf{Poly}(X, [a_0, \ldots, a_n])$ to denote the polynomial $\sum_{j=0}^n a_j X^j$:
\begin{lstlisting}[numbers=none]
>>> f.optimize(), g.optimize()
(Poly((A - B), [0.0, 1.0, 0.5]), Poly((A - B), [0.0, 1.0, -0.5]))
>>> (f * g).optimize()  # == L.optimize()
Poly((A - B), [0.0, 0.0, 1.0, 0.0, -0.25])
\end{lstlisting}

This list of polynomial coefficients is the input required by the quantum singular value transformation (QSVT, \Cref{eq:qsvt}), which can give efficient circuits for suitable polynomials.
\LangName{} first checks and re-scales the coefficients to satisfy the preconditions of \Cref{eq:qsvt}. It next invokes a solver to compute the QSVT rotation angles. Finally, it produces the circuit in \Cref{fig:qsp-opt}.
Because $L$ has degree 4 and even parity, this circuit makes 4 queries to $\oracle{A} - \oracle{B}$ for a total of 8 queries.

The biggest gain, moreover, is hidden. In the LCU-based implementation from \Cref{fig:qsp-circuit}, subnormalization is the value of the un-simplified $L$ with absolute values taken for coefficients:
\[
  \alpha_L^{\text{LCU}} = \left((\abs{\alpha_A} + \abs{-\alpha_B}) + \abs*{\textstyle\frac{1}{2}}{(\abs{\alpha_A} + \abs{-\alpha_B})}^2\right)\left((\abs{\alpha_A} + \abs{-\alpha_B}) + \abs*{-\textstyle\frac{1}{2}}{(\abs{\alpha_A} + \abs{-\alpha_B})^2}\right) = 16,
\]
assuming $\alpha_A = \alpha_B = 1$.
With the QSVT in \Cref{fig:qsp-opt}, it is the \emph{maximum norm} of the simplified $L$:
\[
  \alpha_L^{\text{QSVT}} = \max_{-1 \le x \le 1}\: \abs*{{((\alpha_A + \alpha_B) x)}^2 - \textstyle\frac{1}{4}{((\alpha_A + \alpha_B) x)}^4} = 1,
\]
as \LangName{} computes automatically. Total cost reduces from $12 \times 16 = 192$ to $8 \times 1 = 8$, a $24\times$ speedup.

\section{Language} \label{sec:language}

In this section, we present the \LangName{} language.
First, we formalize a core syntax of mathematical operators to manipulate block-encoded matrices, along with its type system and semantics.
We then extend this core with a symbolic term that enables the polynomial fusion optimization.

\subsection{Core Syntax} \label{sec:core-language}

The core syntax of \LangName{} consists of arithmetic operators over block-encoded matrices:
\begin{align*}
    \tau \Coloneqq{} & \texttt{bool} \mid \tau_1 \otimes \tau_2 \\
    M \Coloneqq{} & \mathcal{B}[A] \mid M^\dagger \mid \lambda_1 M_1 + \lambda_2 M_2 \mid M_1 \cdot M_2 \mid M_1 \oplus M_2 \mid M_1 \otimes M_2 \qquad \abs{\lambda_1} + \abs{\lambda_2} > 0
\end{align*}
where the condition $\abs{\lambda_1} + \abs{\lambda_2} > 0$ prevents division by zero in the denominator of \Cref{eq:prepare}.

\paragraph{Types}
\LangName{} has Booleans and tensor products, where $\texttt{bool}^{\otimes n}$ denotes a tuple of $n$ bits. Semantically, an expression has type $\texttt{bool}^{\otimes n}$ if it encodes an $n$-qubit matrix of dimension $2^n \times 2^n$. Cobble's type system enforces conditions needed for physical realizability of the generated circuit, such as hermiticity for QSVT (\Cref{sec:symbolic-polynomials}) and consistent subnormalization for direct sums (\Cref{sec:core-language-costs}).

\paragraph{Expressions}
The expression $\mathcal{B}[A]$ denotes a black-box block encoding of the matrix $A$.
When $A$ is a standard named gate (e.g., $X$, $Y$), the compiler instantiates $\mathcal{B}[A]$ as that gate with no user action. When $A$ is a compound expression (e.g., $B + C$), the compiler recursively decomposes $\mathcal{B}[A]$ into subexpressions. When $A$ is a basic matrix that does not correspond to a single standard gate (e.g., $\oracle{U_A}$ and $\oracle{U_B}$ in \Cref{sec:regression-example}), \LangName{} expects the user to provide the circuit implementing $\mathcal{B}[A]$. Otherwise, the compiler emits a circuit with unevaluated gate names for the user to define.

Other expressions include adjoints, sums, products, and tensor products of block encodings. The choice operator $\oplus$ denotes a \emph{direct sum} of $M_1$ and $M_2$ that encodes the two matrices in subspaces distinguished by a Boolean, analogous to a conventional sum type or \texttt{if}-expression.

For clarity, the formal syntax and semantics here present binary ($+$, $\,\cdot\,$) rather than $n$-ary ($\sum$, $\prod$) versions of arithmetic operators.
The full version of \LangName{} implemented and studied in subsequent sections provides generalizations to $n$-ary operators, which we briefly describe below.

\paragraph{Type System}

\Cref{fig:core-types} presents the typing rules for the core language of \LangName{}.
A black-box block encoding of a matrix with dimension $2^n \times 2^n$ has type $\texttt{bool}^{\otimes n}$.
The adjoint of an expression has the same type as the original expression. A sum or product of expressions has the type of the summands or factors, provided they have the same type. A direct sum has a Boolean discriminator followed by the type of the summands, subject to a side condition defined in \Cref{sec:core-language-costs} stating that the summands have equal subnormalization. Finally, a tensor product has product type.

\begin{figure}
{\centering\hfuzz=9999pt\resizebox{\textwidth}{!}{\parbox{1.12\textwidth}{%
\begin{mathpar}
\inferrule[T-Base]{A \in \mathbb{R}^{2^n \times 2^n}}{\mathcal{B}[A] : \texttt{bool}^{\otimes n}} \quad
\inferrule[T-Adj]{M : \tau}{M^\dagger : \tau} \quad
\inferrule[T-Sum]{M_1 : \tau \quad M_2 : \tau}{\lambda_1 M_1 + \lambda_2 M_2 : \tau} \quad
\inferrule[T-Product]{M_1 : \tau \quad M_2 : \tau}{M_1 \cdot M_2 : \tau} \quad
\inferrule[T-Choice]{M_1 : \tau \quad M_2 : \tau \quad \alpha_1 = \alpha_2}{M_1 \oplus M_2 : \texttt{bool} \otimes \tau} \quad
\inferrule[T-Tensor]{M_1 : \tau_1 \quad M_2 : \tau_2}{M_1 \otimes M_2 : \tau_1 \otimes \tau_2}
\end{mathpar}%
}}}
\caption{Type system of the core language of \LangName{}. The values $\lambda_1$, $\lambda_2$ are real scalar literals. In \textsc{T-Choice}, $\alpha_1$ and $\alpha_2$ are the subnormalizations of $M_1$ and $M_2$; the side condition $\alpha_1 = \alpha_2$ is defined in \Cref{sec:core-language-costs}.} \label{fig:core-types}
\end{figure}

\subsection{Semantics} \label{sec:denotational-semantics}

Each well-typed program has an abstract denotational semantics giving the matrix encoded by the program and a concrete compilation semantics giving the circuit that realizes the program.

\paragraph{Denotational Semantics}
\newcommand{\matsem}[1]{\llbracket {#1} \rrbracket}
The denotation $\matsem{M}$ is the matrix $M$ encodes, up to subnormalization:
\begin{alignat*}{2}
    \matsem{\mathcal{B}[A]} &= A
    &\qquad
    \matsem{M^\dagger} &= \matsem{M}^\dagger \\
    \matsem{\lambda_1 M_1 + \lambda_2 M_2} &= \lambda_1 \matsem{M_1} + \lambda_2 \matsem{M_2}
    &\qquad
    \matsem{M_1 \cdot M_2} &= \matsem{M_1} \cdot \matsem{M_2} \\
    \matsem{M_1 \oplus M_2} &= \matsem{M_1} \oplus \matsem{M_2}
    &\qquad
    \matsem{M_1 \otimes M_2} &= \matsem{M_1} \otimes \matsem{M_2}
\end{alignat*}
where $\otimes$ is the Kronecker product and $\oplus$ is the direct sum of matrices, $A \oplus B = \begin{bsmallmatrix} A & 0 \\ 0 & B \end{bsmallmatrix}$.

\paragraph{Compilation Semantics}
\newcommand{\circsem}[1]{\llparenthesis {#1} \rrparenthesis}
The circuit $\circsem{M}$ is the sequence of logic gates that realizes $M$ in hardware.
This circuit operates over two registers, following \Cref{eq:block-encoding}: the $n$-qubit data vector $\ket{x}$ to be multiplied by the encoded matrix and the $m$-qubit ancilla $\ket{a}$ to be post-selected to all zeroes.

\Cref{fig:core-circuit} presents the circuit for each operator.
The circuit for a black-box block encoding is the block encoding itself, and the circuit for the adjoint of $M$ is the adjoint of the circuit for $M$.

The circuit for addition uses the LCU method of \citet{childs2012} given in \Cref{def:lcu}. To prepare a superposition of the two branches, the circuit performs a rotation on one ancilla qubit by the angle $\theta = 2\cos^{-1} \smash{\sqrt{\abs{\lambda_1} / (\abs{\lambda_1} + \abs{\lambda_2})}}$. The $n$-ary case would use $\lceil \log n \rceil$ ancilla qubits and a set of controlled rotations. For $\lambda_j < 0$, the circuit adds a $Z$ gate controlled on that branch.

The circuit for multiplication sequentially executes each factor in the conventional reverse order. Following \citet{sunderhauf2023,dalzell2025}, it reuses ancillas between factors and adds one ancilla $\ket{a_0}$ to ensure that all intermediate ancilla states after each factor are post-selected to zero. The $n$-ary case would add $\lceil \log n \rceil$ ancillas, replace the anti-controlled NOT between factors by an anti-controlled integer increment, and replace the final NOT with a subtraction by $n - 1$.

The circuit for direct sum selects between the circuits for the branches based on the discriminator bit $\ket{x_0}$. It parallels the case for ordinary sum but considers the discriminator as part of the data, which will not be post-selected, rather than the ancilla, which will.
Finally, the circuit for tensor product independently executes a circuit for each factor on its data and ancilla components.

\begin{figure}
\captionsetup[subfigure]{labelformat=empty}
\captionsetup{justification=centering}
\centering
\begin{subfigure}[b]{.2\textwidth}%
\centering%
\scalebox{0.75}{%
\tikzsetfigurename{circsem-base}%
\begin{quantikz}[classical gap=0.075cm, row sep={0.85cm, between origins}]
\lstick{$\ket{a}$} \setwiretype{b} & \gate[2]{\vphantom{{}^\dagger}\mathcal{B}[A]} & \rstick{$\ket{a'}$} \\
\lstick{\hphantom{'}$\ket{x}$} \setwiretype{b} & & \rstick{$\ket{x'}$}
\end{quantikz}%
}%
\caption{$\vphantom{{}^\dagger}\circsem{\mathcal{B}[A]}$}%
\end{subfigure}\quad%
\begin{subfigure}[b]{.2\textwidth}%
\centering%
\scalebox{0.75}{%
\tikzsetfigurename{circsem-adjoint}%
\begin{quantikz}[classical gap=0.075cm, row sep={0.85cm, between origins}]
\lstick{$\ket{a}$} \setwiretype{b} & \gate[2]{\circsem{M}^\dagger} & \rstick{$\ket{a'}$} \\
\lstick{\hphantom{'}$\ket{x}$} \setwiretype{b} & & \rstick{$\ket{x'}$}
\end{quantikz}%
}%
\caption{$\circsem{M^\dagger}$}%
\end{subfigure}\ \ %
\begin{subfigure}[b]{.52\textwidth}%
\centering%
\scalebox{0.75}{%
\tikzsetfigurename{circsem-sum}%
\begin{quantikz}[classical gap=0.075cm, row sep={0.85cm, between origins}]
\lstick{$\ket{a_0}$} & \gate{R_y(\theta)}\gategroup[1,steps=1,style={draw=none, inner ysep=0pt}, label style={label position=south,anchor=north}]{\footnotesize\textcolor{gray}{\textsc{Prepare}\vphantom{${}^\dagger$}}} & \octrl{1}\gategroup[2,steps=2,style={draw=none, inner ysep=0pt}, label style={label position=north,anchor=north,yshift=-0.22cm}]{\footnotesize\textcolor{gray}{\textsc{Select}}} & \ctrl{1} & \gate{R_y(-\theta)}\gategroup[1,steps=1,style={draw=none, inner ysep=0pt}, label style={label position=south,anchor=north}]{\footnotesize\textcolor{gray}{\textsc{Prepare}${}^\dagger$}} & \rstick{$\ket{a_0'}$} \\
\lstick{\hspace{1.2em}$\ket{a, x}$} \setwiretype{b} & & \gate{\circsem{M_1}} & \gate{\circsem{M_2}} & & \rstick{$\ket{a', x'}$}
\end{quantikz}%
}%
\caption{$\vphantom{{}^\dagger}\circsem{\lambda_1 M_1 + \lambda_2 M_2}$}%
\end{subfigure}\\[1.75ex]
\begin{subfigure}[b]{.32\textwidth}%
\centering%
\scalebox{0.75}{%
\tikzsetfigurename{circsem-product}%
\begin{quantikz}[classical gap=0.075cm, row sep={0.85cm, between origins}]
\lstick{$\ket{a_0}$} & & \targ{} & \gate{X} & \rstick{$\ket{a_0'}$} \\
\lstick{$\ket{a}$} \setwiretype{b} & \gate[2]{\circsem{M_2}} & \octrl{-1} & \gate[2]{\circsem{M_1}} & \rstick{$\ket{a'}$} \\
\lstick{$\ket{x}$} \setwiretype{b} & & & & \rstick{$\ket{x'}$}
\end{quantikz}%
}%
\caption{$\circsem{M_1 \cdot M_2}$}%
\end{subfigure}\quad\ %
\begin{subfigure}[b]{.28\textwidth}%
\centering%
\scalebox{0.75}{%
\tikzsetfigurename{circsem-choice}%
\begin{quantikz}[classical gap=0.075cm, row sep={0.85cm, between origins}]
\lstick{$\ket{x_0}$} & \octrl{1} & \ctrl{1} & \rstick{$\ket{x_0'}$} \\
\lstick{$\ket{a}$} \setwiretype{b} & \gate[2]{\circsem{M_1}} & \gate[2]{\circsem{M_2}} & \rstick{$\ket{a'}$} \\
\lstick{$\ket{x}$} \setwiretype{b} & & & \rstick{$\ket{x'}$}
\end{quantikz}%
}%
\caption{$\circsem{M_1 \oplus M_2}$}%
\end{subfigure}\quad\ %
\begin{subfigure}[b]{.24\textwidth}%
\centering%
\scalebox{0.75}{%
\tikzsetfigurename{circsem-tensor}%
\begin{quantikz}[classical gap=0.075cm, row sep={0.85cm, between origins}]
\setwiretype{b} \lstick{$\ket{a_1, x_1}$} & \gate{\circsem{M_1}} & \rstick{$\ket{a_1', x_1'}$}\\
\setwiretype{b} \lstick{$\ket{a_2, x_2}$} & \gate{\circsem{M_2}} & \rstick{$\ket{a_2', x_2'}$}
\end{quantikz}%
}%
\caption{$\circsem{M_1 \otimes M_2}$}%
\end{subfigure}
\caption{Compilation semantics of the core language of \LangName{}.} \label{fig:core-circuit}
\end{figure}

\paragraph{Soundness}
The type system is sound with respect to both semantics. Furthermore, the denotational semantics is equal to the sub-matrix in the top-left block of the matrix representation of the compilation semantics, up to rescaling by the subnormalization $\alpha \in \mathbb{R}$ (\Cref{def:block-encoding}).

\begin{theorem}
    If $M: \tau$, then $\matsem{M}$ is a valid matrix and $\circsem{M}$ is a valid quantum circuit.
\end{theorem}
\begin{proof}
    By induction on the structure of $M$.
\end{proof}

\begin{theorem}
    Assume that $M: \textup{\texttt{bool}}^{\otimes n}$ and for all $\mathcal{B}[A] : \textup{\texttt{bool}}^{\otimes n_A}$ in $M$, the $2^{n_A} \times 2^{n_A}$ top-left block of $\mathcal{B}[A]$ is $A / \alpha_A$ for some $\alpha_A$. Then, the $2^n \times 2^n$ top-left block of $\circsem{M}$ is $\matsem{M} / \alpha$ for some $\alpha$.
\end{theorem}
\begin{proof}
    By induction on the structure of $M$, and invoking the cited prior results.
\end{proof}

\subsection{Symbolic Polynomials} \label{sec:symbolic-polynomials}
We next extend the core language with a symbolic term in the compiler's intermediate representation that captures matrix polynomials. The compiler rewrites the AST using this term so it can check degree and parity of polynomials and compute QSVT angles to generate a circuit. This term is:
\begin{align*}
    M \Coloneqq{} & \cdots \mid \textsf{Poly}(M, [a_0, \ldots, a_d]) \qquad a_j \in \mathbb{R}
\end{align*}

The denotational semantics of $\textsf{Poly}(M, [a_0, \ldots, a_d])$ is defined to be equal to $\matsem{\smash{\sum_{j=0}^d a_j M^j}}$ in the core language. But when $M$ is Hermitian, i.e.\ $M = M^\dagger$, the polynomial can be compiled to a more efficient circuit using the quantum singular value transformation as defined in \Cref{eq:qsvt}.

\paragraph{Type System}
In \Cref{fig:types-poly}, we present rules that augment the type system to conservatively check hermiticity of $M$. For a black-box block encoding $\mathcal{B}[A]$, \LangName{} requires the user to specify whether $A$ is Hermitian. The adjoint of a Hermitian matrix is Hermitian, as are the sums, direct sums, and tensor products of Hermitian matrices. Products of Hermitian matrices are Hermitian if and only if the factors commute, which \LangName{} also requires the user to specify.
Finally, symbolic polynomials with real coefficients in $M$ are well-typed and Hermitian when $M$ is Hermitian.

\begin{figure}
{\centering\hfuzz=9999pt\resizebox{\textwidth}{!}{\parbox{1.12\textwidth}{%
\begin{mathpar}
\inferrule[H-Base]{\textcolor{gray}{A = A^\dagger}}{\mathcal{B}[A] = (\mathcal{B}[A])^\dagger}

\inferrule[H-Adj]{M = M^\dagger}{M^\dagger = M}

\inferrule[H-Sum]{M_1 = M_1^\dagger \quad M_2 = M_2^\dagger}{M_1 + M_2 = (M_1 + M_2)^\dagger}

\inferrule[H-Product]{\textcolor{gray}{\matsem{M_1} \matsem{M_2} = \matsem{M_2} \matsem{M_1}} \quad M_1 = M_1^\dagger \quad M_2 = M_2^\dagger}{M_1 \cdot M_2 = (M_1 \cdot M_2)^\dagger}

\inferrule[H-Choice]{M_1 = M_1^\dagger \quad M_2 = M_2^\dagger}{M_1 \oplus M_2 = (M_1 \oplus M_2)^\dagger}

\inferrule[H-Tensor]{M_1 = M_1^\dagger \quad M_2 = M_2^\dagger}{M_1 \otimes M_2 = (M_1 \otimes M_2)^\dagger}

\inferrule[H-Poly]{M = M^\dagger}{\textsf{Poly}(M, p) = \textsf{Poly}(M, p)^\dagger}

\inferrule[T-Poly]{M : \tau \quad M = M^\dagger}{\textsf{Poly}(M, p) : \tau}
\end{mathpar}%
}}}
\caption{Typing rules to check hermiticity and symbolic polynomials. Conditions in \textcolor{gray}{gray} are user-provided.} \label{fig:types-poly}
\end{figure}

\paragraph{Compilation Semantics}

Any polynomial $\textsf{Poly}(M, [a_0, a_1, \ldots])$ can be decomposed into even and odd parts $\textsf{Poly}(M, [a_0, 0, a_2, \ldots]) + \textsf{Poly}(M, [0, a_1, 0, a_3, \ldots])$.
The QSVT provides an efficient circuit for each part \citep{lin2022}, depicted in \Cref{fig:qsvt-circuit}. In the circuit, the phase angles $\phi_j$ are computed from the $a_j$ using a quantum signal processing (QSP) solver such as pyQSP \citep{martyn2021}.

\begin{figure}
\centering
\resizebox{\textwidth}{!}{
\tikzsetfigurename{qsvt-circuit}%
\begin{quantikz}[classical gap=0.075cm, row sep={0.85cm, between origins}]
\lstick{$\ket{a_0}$} & \gate{H} & \targ{} & \gate{R_z(\phi_0)} & \targ{} & & \targ{} & \gate{R_z(\phi_1)} & \targ{} & & \ \ldots\  & & \targ{} & \gate{R_z(\phi_{d-1})} & \targ{} & & \targ{} & \gate{R_z(\phi_d)} & \targ{} & \gate{H} & \rstick{$\ket{a_0}$} \\
\lstick{$\ket{a}$} \setwiretype{b} & & \octrl{-1} & & \octrl{-1} & \gate[2]{\circsem{M}} & \octrl{-1} & & \octrl{-1} & \gate[2]{\circsem{M}^\dagger} & \ \ldots\  & \gate[2]{\circsem{M}} & \octrl{-1} & & \octrl{-1} & \gate[2]{\circsem{M}^\dagger} & \octrl{-1} & & \octrl{-1} & & \rstick{$\ket{a}$} \\
\lstick{$\ket{x}$} \setwiretype{b} & & & & & & & & & & \ \ldots\  & & & & & & & & & & \rstick{$\ket{x}$}
\end{quantikz}}
\caption{Compilation of a degree-$d$ polynomial using QSVT\@. Even $d$ is shown; the odd case ends on $\circsem{M}$.} \label{fig:qsvt-circuit}
\end{figure}

\section{Cost Model} \label{sec:cost-model}

In this section, we present a cost model provided as an analysis in \LangName{} that enables the developer to estimate query and subnormalization costs of programs. Building on these principles, we analytically compare the efficiency of different approaches to realize matrix polynomials.
The cost model informs optimization decisions in two ways: determining which factoring transformations improve cost and guiding a fallback to other methods should QSVT incur excessive query costs.

\subsection{Costs of Core Language} \label{sec:core-language-costs}

In \Cref{tbl:core-language-costs}, we summarize the runtime costs of each operator in the core language of \LangName{}.
This table combines and generalizes prior results from the theoretical literature, in particular those of \citet{gilyen2019,lin2022,harrigan2024,dalzell2025}.

\paragraph{Queries}
The first cost is the number of queries to black-box block encoding oracles, examples of which are $\oracle{X}$ and $\oracle{Y}$ from \Cref{sec:simulation-example} and $\oracle{U_A}$ and $\oracle{U_B}$ from \Cref{sec:regression-example}.
This quantity is proportional to the precise total number of logic gates but less subject to implementation variance.
By definition, a black-box block encoding makes one query. Adjoints make the same number of queries as the original expression. For all other operators, the number of queries is the total of the operands.

\paragraph{Subnormalization}
The second cost is the subnormalization $\alpha$ of \Cref{def:block-encoding} that scales the encoded matrix and is proportional to the expected number of circuit repetitions needed to produce the matrix.
Black-box terms have subnormalization specified by the user.
Adjoints have subnormalization equal to the original expression.
Following \Cref{def:lcu}, the subnormalization of a sum is the sum of those of the summands, weighted by the absolute value of the coefficients.
For products and tensor products, the subnormalization is the product of those of the operands.

Revisiting the rule \textsc{T-Choice} for direct sums in \Cref{fig:core-types}, a direct sum is only well-defined when the operands have equal subnormalization; the direct sum then takes on that subnormalization.
The reason is that $(A / \alpha_1) \oplus (B / \alpha_2) = (A \oplus B) / \alpha$ for some $\alpha$ only when $\alpha_1 = \alpha_2 = \alpha$, a restriction unique to direct sums.
The side condition in \textsc{T-Choice} thus ensures that a direct sum is well-defined, and the \LangName{} type checker automatically checks this condition by computing $\alpha$ via \Cref{tbl:core-language-costs}.

\paragraph{Qubits}
The third cost is the number of ancilla qubits required to implement the operator.
For black-box block encodings, this number is user-specified.
For adjoints, it is the same as the original expression.
For the other operators, it is the maximum of the number of qubits of the operands.
Sums and products also introduce $\lceil \log n \rceil$ selection bits, as discussed in \Cref{sec:denotational-semantics}.

\paragraph{Circuit Costs}
Though not listed in the table, the total gate count of the computation is the product of the subnormalization, number of queries, and number of gates per query.
The last term depends on the implementation of each $\mathcal{B}[A]$, and provides information orthogonal to the other dominant costs in the table.
We report concrete gate and qubit counts empirically in \Cref{sec:eval-comparison}.

\begin{table}
\caption{Query, subnormalization, and ancilla costs of the core language of \LangName{}.} \label{tbl:core-language-costs}
\begin{tabular}{c c c c}
\toprule
Operator & \# Queries $k$ & Subnormalization $\alpha$ & \# Qubits $m$ \\
\midrule
$\mathcal{B}[A]$ & $1$ & user-specified & user-specified \\[0.5ex]
$M^\dagger$ & $k_M$ & $\alpha_M$ & $m_M$ \\[0.5ex]
$\sum_{j=1}^n \lambda_j M_j$ & $\sum_{j=1}^n k_j$ & $\sum_{j=1}^n \abs{\lambda_j} \alpha_j$ & $\lceil \log n \rceil + \max_j m_j$ \\[0.5ex]
$\prod_{j=1}^n M_j$ & $\sum_{j=1}^n k_j$ & $\prod_{j=1}^n \alpha_j$ & $\lceil \log n \rceil + \max_j m_j$ \\[0.5ex]
$\bigoplus_{j=1}^n M_j$ & $\sum_{j=1}^n k_j$ & $\alpha_j$ (all equal) & $\max_j m_j$ \\[0.5ex]
$\bigotimes_{j=1}^n M_j$ & $\sum_{j=1}^n k_j$ & $\prod_{j=1}^n \alpha_j$ & $\max_j m_j$ \\[0.25ex]
\bottomrule
\end{tabular}
\end{table}

\subsection{Costs of Polynomials} \label{sec:polynomials}

\begin{table}
\vspace*{1ex}%
\caption{Worst-case (mixed-parity) costs of implementations of $\textsf{Poly}(M, p)$ where $p = [a_0, \ldots, a_d]$.}\label{tbl:polynomial-costs}
\begin{tabular}{c c c c}
\toprule
Method & \# Queries $k$ & Subnormalization $\alpha$ & \# Qubits $m$ \\
\midrule
LCU & $k_M \sum_{j \mid a_j \neq 0} j$ & $\norm{p(\alpha_M x)}_1$ & $\lceil \log d \rceil + d + m_M$ \\[0.5ex]
Horner & $k_M d$ & $\norm{p(\alpha_M x)}_1$ & $2d + m_M$ \\[0.5ex]
QSVT & $k_M (2d - 1)$ & $\norm{p_\text{even}(\alpha_M x)}_\infty + \norm{p_\text{odd}(\alpha_M x)}_\infty$ & $2 + m_M$ \\[0.5ex]
GQET & $k_M d$ & $\norm{T_p(\alpha_M x)}_\infty$ & $1 + m_M$ \\[0.25ex]
\bottomrule
\end{tabular}
\end{table}

In \Cref{tbl:polynomial-costs}, we summarize the cost of implementing $\textsf{Poly}(M, [a_0, \ldots, a_d])$ by four different methods:
\begin{itemize}
    \item Linear combination of unitaries (LCU, \Cref{def:lcu}), which directly evaluates the sum of monomials $\smash{\sum_{j=0}^d a_j M^j}$ using the core language operators and costs given in \Cref{tbl:core-language-costs}.
    \item Horner's method, which decomposes $\smash{\sum_{j=0}^d a_j M^j} = (((a_d M + a_{d-1} I)M + \cdots)M + a_0 I)$ and evaluates the polynomial by $d$ iterations of multiplication and addition.
    \item Quantum singular value transformation (QSVT, \Cref{eq:qsvt}), which constructs the even and odd parts of the polynomial each using the circuit in \Cref{fig:qsvt-circuit} and takes their sum.
    \item Generalized quantum eigenvalue transformation (GQET) \citep{sunderhauf2023}, which extends QSVT to mixed-parity polynomials by replacing $R_z$ with arbitrary rotations.
\end{itemize}

\paragraph{LCU vs. Horner}
As shown in the table, the direct implementation by LCU typically requires more queries than the other methods. It makes $j$ queries to $M$ for each monomial $M^j$ in the sum, whereas Horner's method only makes one query for each of the $d$ iterations. Horner's method, however, suffers a penalty by incurring two ancillas per iteration -- one for the multiplication and one for the addition -- whereas LCU performs one sum at the end with only logarithmic cost.

Both methods have the same subnormalization, the $\ell_1$ norm of the polynomial coefficients:
\[
\norm{p(\alpha_M x)}_1 = \sum_{j=0}^d\, \abs*{a_j \alpha_M^j}\,,
\]
and is not affected by the difference in order and factoring of arithmetic operations.

\paragraph{LCU vs. QSVT}
Typically, QSVT is more efficient than LCU\@. For a polynomial with mixed parity $\smash{\sum_{j=0}^d a_j M^j}$ where $a_d$ and $a_{d-1}$ are both nonzero, the sum of even and odd parts by QSVT makes $d + (d - 1)$ total queries to $M$, whereas LCU makes that many for $M^d$ and $M^{d-1}$ alone. Moreover, QSVT uses fewer ancillas --- one for the circuit in \Cref{fig:qsvt-circuit} and one for the final sum.

Taking the simplifying assumption that $p$ has fixed parity, the subnormalization for QSVT is equal to the $L_\infty$ (uniform) norm of the polynomial, which is no greater than the $\ell_1$ norm:
\[
\norm{p(\alpha_M x)}_\infty = \max_{-1 \le x \le 1}\: \abs*{\sum_{j=0}^d a_j \alpha_M^j x^j} \le \norm{p(\alpha_M x)}_1
\]
by the triangle inequality. It can be much smaller when coefficients $a_j$ have mixed signs.

\paragraph{QSVT vs. GQET}
For conceptual completeness, we also compare against the GQET, which generalizes QSVT to mixed-parity polynomials without the need to explicitly split into even and odd parts. It incurs subnormalization equal to the $L_\infty$ norm of the following modified polynomial:
\[
\norm{T_p(\alpha_M x)}_\infty = \max_{\abs{z} = 1}\: \abs*{\sum_{j=0}^d a_j T_j(\alpha_M z)} \le O(\log d) \cdot \norm{p(\alpha_M x)}_\infty\,,
\]
where $z \in \mathbb{C}$ and $T_j(x)$ is a Chebyshev polynomial of the first kind. This function is non-trivial to compare against the previous cases; \citet{sunderhauf2023} proves the asymptotic bound above.

\subsection{Soundness and Implications} \label{sec:cost-implications}

The cost model accurately predicts the costs of the core language and polynomials in \LangName{}:

\begin{theorem}
\label{thm:soundness}
A well-typed program compiles to a circuit with costs given by \Cref{tbl:core-language-costs,tbl:polynomial-costs}.
\end{theorem}
\begin{proof}
By induction on the structure of the program. Query and ancilla counts follow directly from the circuits in \Cref{fig:core-circuit,fig:qsvt-circuit}.
Subnormalization for sums, products, and tensor products is proven by \citet{gilyen2019}.
Subnormalization for polynomials by LCU and Horner follows by induction.
For QSVT, subnormalization follows from the conditions on $P(x)$ in \Cref{eq:qsvt}.
\end{proof}

The cost model also offers a convenient way to analyze the effect -- or lack thereof -- of refactoring operators in matrix expressions.
Horner's method refactors the additions and multiplications in a polynomial, which eliminates redundant queries but cannot change the subnormalization.

\Cref{sec:regression-example} illustrates how more general instances of subexpression reuse do not lead to speedup. Absent additional structure of $A$ and $B$, the expression $(A + B) \cdot (A + B)$ requires two additions, one multiplication, and four total queries to $A$ and $B$\@.
Exponentiation by squaring is also not admissible in general: $A^{128} = A^{64} \cdot A^{64}$, but squaring $A^{64}$ costs the same as multiplying it by $A$ for 64 times.

\section{Optimizations} \label{sec:optimizations}

In this section, we present the optimizations of sum fusion and polynomial fusion in \LangName{}, along with a set of additional rewrites that enable and complement these optimizations.
We show that the system overall is sound, strongly normalizing, and cost-nonincreasing.

\subsection{Sum Fusion}

The overarching principle of sum fusion is to flatten nested linear combinations of expressions to eliminate intermediate overhead from subnormalization:
\begin{equation*}
\sum_{k} \left( \sum_{j} a_{k,j} \mathcal{B}[M_j] \right) \mapsto \sum_{j} \left(\sum_k a_{k,j}\right) \mathcal{B}[M_j]
\end{equation*}

\paragraph{Soundness}
Sum fusion preserves the block encoding semantics of the expression, as can be seen by direct algebraic simplification.
Note that it does not strictly preserve the compilation semantics.

\paragraph{Cost Reduction}
When all coefficients are positive, sum fusion leaves queries and subnormalization unchanged and modestly reduces ancilla count.
But when some signs are negative, it can cancel queries and reduce subnormalization by the triangle inequality: $\sum_{j} \abs{\sum_k a_{k,j}} \le \sum_{k} \sum_{j} \abs{a_{k,j}}$.

The compilation of sums benefits from another practical optimization of merging the subnormalization of each sub-expression with its coefficient.
For example, consider the expression $A + 100\, B$ with $\alpha_A = 100$ and $\alpha_B = 1$.
Then, no rotation to prepare $[\smash{\sqrt{1}}, \smash{\sqrt{100}}]$ is needed because simply adding the encodings of $A$ and $B$ gives the correct weighted sum where $\alpha$ effectively scales $B$ by 100.

\subsection{Polynomial Fusion}

The overarching principle of polynomial fusion is to merge monomials with the same base expression into symbolic terms that enable more efficient implementation by QSVT\@:
\begin{equation*}
\sum_{j=0}^d a_j M^j \mapsto \textsf{Poly}(M, [a_0, \ldots, a_d])
\end{equation*}

\paragraph{Soundness}
Like sum fusion, polynomial fusion preserves the block encoding semantics of the expression (by definition of \textsf{Poly}) but does not strictly preserve the compilation semantics.

\paragraph{Cost Reduction}
For all fixed-parity polynomials, fusion into \textsf{Poly} and implementation by QSVT reduce the number of queries to the degree $d$ of the polynomial, whereas in LCU it is greater than $d$ for non-monomials.
As shown in \Cref{sec:polynomials}, subnormalization reduces to the $L_\infty$ norm, which is no greater than the $\ell_1$ norm as in LCU and can be strictly less for mixed-sign coefficients.

QSVT must separate mixed-parity polynomials into even and odd parts, which can increase the number of queries by up to a factor of two over Horner's method in principle.
As we show in \Cref{sec:evaluation}, the improved subnormalization typically outweighs the cost of these queries.
\LangName{} falls back to LCU or Horner's method otherwise, ensuring that cost is nonincreasing overall.

\subsection{Additional Transformations}

\begin{figure}
\def\MathparLineskip{\lineskip=0.3ex}%
\begin{mathpar}
\textsf{Poly}(A, f) \cdot \textsf{Poly}(A, g) \mapsto \textsf{Poly}(A, f \cdot g)

\textsf{Poly}(A, f) + \textsf{Poly}(A, g) \mapsto \textsf{Poly}(A, f + g) \vspace{1.55ex}

\textsf{Poly}(A, f) \oplus \textsf{Poly}(B, f) \mapsto \textsf{Poly}(A \oplus B, f)

\textsf{Poly}(\textsf{Poly}(A, f), g) \mapsto \textsf{Poly}(A, g \circ f) \vspace{1.55ex}

(A \cdot B) + (A \cdot C) \mapsto A \cdot (B + C)

(B \cdot A) + (C \cdot A) \mapsto (B + C) \cdot A \vspace{1.55ex}

(A \cdot B) \oplus (A \cdot C) \mapsto (I \otimes A) \cdot (B \oplus C)

(B \cdot A) \oplus (C \cdot A) \mapsto (B \oplus C) \cdot (I \otimes A) \vspace{1.55ex}

(A \otimes B) + (A \otimes C) \mapsto A \otimes (B + C)

(B \otimes A) + (C \otimes A) \mapsto (B + C) \otimes A \vspace{1.55ex}

(A \otimes B) \oplus (A \otimes C) \mapsto A \otimes (B \oplus C)

(B \otimes A) \oplus (C \otimes A) \mapsto (B \oplus C) \otimes A
\end{mathpar}\vspace*{-1.7ex}%
\begin{mathpar}
A \oplus A \mapsto I \otimes A

A \cdot I \mapsto A

A^{\dagger} \mapsto A \quad \text{(when $A = A^\dagger$)}

(A^{\dagger})^{\dagger} \mapsto A \vspace{1.55ex}

(A \cdot B)^{\dagger} \mapsto B^{\dagger} \cdot A^{\dagger}

(A + B)^{\dagger} \mapsto A^{\dagger} + B^{\dagger}

(A \otimes B)^{\dagger} \mapsto A^{\dagger} \otimes B^{\dagger}

(A \oplus B)^{\dagger} \mapsto A^{\dagger} \oplus B^{\dagger}
\end{mathpar}
\caption{Selection of additional rewrites that enable and complement sum and polynomial fusion.} \label{fig:additional-transformations}
\end{figure}

In \Cref{fig:additional-transformations}, we present additional transformations that either expose more opportunities to apply sum and polynomial fusion or eliminate redundant queries from the program.
All of the rewrites in the figure are sound by algebraic reasoning on the denotational semantics of \LangName{}.
They never increase the number of queries, and they keep the subnormalization unchanged except when exposing more opportunities for sum and polynomial fusion.

\paragraph{Polynomials}
The first few rules simplify symbolic polynomials as much as possible. They merge different polynomials with the same base by multiplying or adding the coefficients, merge direct sums of the same polynomial with different bases by taking the direct sum of the base expressions, and merge nested polynomials by composing the functions given by their coefficients.

\paragraph{Factoring}
The next set of rules factor common subexpressions to remove redundant queries.
Most originate from matrix algebra, e.g.\ distributivity of (tensor) products over (direct) sums.

In general, factoring is among the only forms of subexpression reuse that directly improve cost in block encodings. It only applies in limited cases.
The example $(A + B) \cdot (A + B)$ from \Cref{sec:cost-implications} does not factor into fewer instances of $A$, whereas $(A \cdot B) + (A \cdot C) = A \cdot (B + C)$ does.

\paragraph{Simplification}
The last few rules simplify the expression by pushing down adjoints and eliminating constants.
When the type checker finds that a matrix is Hermitian, adjoints can cancel.

\paragraph{Implementation}
In \LangName{}, polynomial discovery is heuristic: the rewrite system does not guarantee that an optimal polynomial form will be found.
That said, under a well-founded measure on expressions that strictly decreases under the rules, the rewrite system is strongly normalizing.
Supposing the following priority order for rewrites, the system reaches a unique normal form.

In our implementation, optimization proceeds bottom-up, such that subexpressions are optimized first. For sums, rewrites apply in the following order: like terms are flattened, common subexpressions are factored, scalar coefficients are merged into polynomials, polynomials with the same base are combined, and constant terms are merged with polynomials. For products, factors are flattened, polynomial product fusion is applied when possible, and repeated factors are collapsed to $\textsf{Poly}$.

\section{Evaluation} \label{sec:evaluation}

We implemented \LangName{} as an embedded language in Python along with a compiler, simulator, and benchmark suite. In this section, we use \LangName{} to answer the following research questions:
\begin{enumerate}[leftmargin=3em]
  \item [\emph{RQ1.}] By how much do the proposed optimizations reduce the cost of matrix expressions?
  \item [\emph{RQ2.}] Can the \LangName{} system empirically analyze the costs of known quantum algorithms?
  \item [\emph{RQ3.}] By how much do existing circuit optimizers reduce the costs of these programs?
  \item [\emph{RQ4.}] How scalable is the \LangName{} compiler in compile time with varying problem size?
\end{enumerate}

\paragraph{Implementation}
Given a program, the compiler performs type checking (\Cref{sec:core-language}), cost analyses (\Cref{sec:cost-model}), and optimizations (\Cref{sec:optimizations}). It then outputs a quantum circuit in the OpenQASM 2.0 \citep{cross2017} format.
To solve for QSP phase angles, the compiler invokes pyQSP \citep{martyn2021,chao2020,dong2021} or optionally PennyLane \citep{bergholm2022}.
For testing, the simulator invokes Quimb \citep{gray2018} to perform classical circuit simulation.

We have released this package as an open-source repository. Moreover, all source code, benchmarks, and experimental scripts are available as part of the software artifact of this paper.

\subsection{RQ1: Optimization of Matrix Expressions}

\begin{enumerate}
\item [\emph{RQ1.}] By how much do the proposed optimizations reduce the cost of matrix expressions?
\end{enumerate}

\paragraph{Benchmarks}
This question requires us to establish a set of executable program or circuit benchmarks for the emerging domain of quantum linear algebra applications. We used \LangName{} to express block encodings of the input matrices corresponding to three applications from the literature:
\begin{itemize}
  \item \emph{Simulation} (\texttt{penalized-coupler}): a Hamiltonian describing a coupled system subject to a penalty function, which resembles simulation of an Ising model \citep{cervera-lierta2018} or the adiabatic optimization of a constrained satisfaction problem \citep{farhi2000}.
  \item \emph{Regression} (\texttt{ols-ridge}): a regularized Gram matrix that interpolates between ordinary least squares \citep{chakraborty2023} and ridge regression \citep{yu2021} for a model.
  \item \emph{Image Processing} (\texttt{laplacian-filter}): a two-dimensional Laplacian stencil over a rectangle, as applied in quantum algorithms for edge detection \citep{fan2019}.
\end{itemize}

Though some of these applications are not realizable on near-term hardware, they are representative of the design space of the inputs and structure of quantum algorithms for linear algebra, and they test the effectiveness of \LangName{} for implementing matrix expressions in different domains.

\paragraph{Cost Metric}
For each benchmark, we computed its runtime cost before and after optimization by the \LangName{} compiler. We define cost using a formula adapted from \citet{sunderhauf2024}:
\[
\text{Cost} = \text{\# Queries} \times \text{Subnormalization},
\]
where the first term is the number of queries to oracles that encode basic matrices, examples of which are \oracle{$X$} and \oracle{$Y$} from \Cref{sec:simulation-example} and \oracle{$U_A$} and \oracle{$U_B$} from \Cref{sec:regression-example}. This number is proportional to the precise count of quantum logic gates in each run of the circuit, which is subject to implementation variance.
The second term, subnormalization, is proportional to the number of runs of the circuit required to successfully produce the target matrix through post-selection (\Cref{sec:qcla}).

\paragraph{Results}
In \Cref{tbl:benchmark-comparison}, we present the cumulative effect of the optimizations of \Cref{sec:optimizations} in terms of runtime cost reduction on each benchmark.
The results indicate meaningful speedups under certain settings, ranging from $2.6\times$ for the simulation examples to $25.4\times$ for the regression examples.
Generally, the speedup is larger for longer programs with higher-degree polynomials.

The reported speedups are attributable to the compiler passes proposed in this work: sum fusion and polynomial fusion (and the related rewrites in \Cref{sec:optimizations}), which restructure expressions that would naively compile to nested LCU into compact forms enabling compilation via QSVT\@.

\begin{table}
\centering
\caption{Runtime cost comparison for matrix expression benchmarks. Units for cost are \# of queries.}
\label{tbl:benchmark-comparison}
\begin{tabular}{ l r r r }
  \toprule
  & \multicolumn{2}{c}{{\# Queries $\times$ Subnormalization = Cost}} & \\
  \cmidrule{2-3}
  {Matrix Expression} & \multicolumn{1}{c}{{Unoptimized}} & \multicolumn{1}{c}{{Optimized}} & {Speedup} \\
  \midrule
  \texttt{simulation-example} (\S~\ref{sec:simulation-example}) & $8 \times 2.6 = 20.8$ & $4 \times 2.0 = 8.0$ & 2.6$\times$ \\
  \texttt{regression-example} (\S~\ref{sec:regression-example}) & $12 \times 16.0 = 192.0$ & $8 \times 1.0 = 8.0$ & 24.0$\times$ \\
  \midrule
  \texttt{penalized-coupler} & $6 \times 8.2 = 49.2$ & $3 \times 6.2 = 18.6$ & 2.6$\times$ \\
  \texttt{laplacian-filter} & $8 \times 59.3 = 474.2$ & $2 \times 27.8 = 55.7$ & 8.5$\times$ \\
  \texttt{ols-ridge} & $148 \times 529.7 = \num{78398.8}$ & $18 \times 171.5 = \num{3087.0}$ & 25.4$\times$ \\
  \bottomrule
\end{tabular}
\end{table}

\subsection{RQ2: Resource Analysis of Quantum Algorithms}

\begin{enumerate}
  \item [\emph{RQ2.}] Can the \LangName{} system empirically analyze the costs of known quantum algorithms?
\end{enumerate}

\paragraph{Benchmarks}
For this question, we used \LangName{} to express block encodings of matrix polynomials that underpin three major quantum algorithms as described by \citet{martyn2021}:
\begin{itemize}
  \item \emph{Matrix inversion}: the optimal polynomial approximation by \citet{sunderhauf2025} for $f(A) = A^{-1}$ in the quantum linear system solver \citep{harrow2009,childs2017}.
  \item \emph{Hamiltonian simulation}: the Jacobi-Anger decomposition of $f(A) = e^{-iAt}$ in the quantum algorithm to solve the time-dependent Schrödinger equation \citep{low2019}.
  \item \emph{Spectral thresholding}: the Fourier-Chebyshev expansion of $f(A) = \mathrm{sign}(A) = A(A^2){}^{-1/2}$ in quantum search, phase estimation, and eigenvalue filtering \citep{martyn2021}.
\end{itemize}

The benchmarks of RQ1 correspond to the input matrices $A$ of these functions $f(A)$.
But unlike the input matrices, these $f(A)$ are already assumed by their authors to utilize QSVT for polynomial evaluation, on the basis that QSVT is more efficient than other methods such as LCU\@.
Our evaluation is thus intended to empirically analyze, not surpass, these theoretical claims of efficiency.

\paragraph{Cost Metric}
We computed the runtime cost (defined as in RQ1) of each benchmark, using both fusions plus different implementations of polynomials. The first implementation is by linear combination of unitaries (\Cref{def:lcu}), and the second is by the quantum singular value transformation (\Cref{eq:qsvt}) as intended by the authors of these algorithms.
Their difference effectively captures the additional speedup contributed by polynomial fusion via QSVT, versus sum fusion alone.

As another point of comparison, we computed the cost of the polynomial evaluated by Horner's method as described in \Cref{sec:polynomials}.
We also evaluated against the GQET from \Cref{sec:polynomials}, but its subnormalization and costs for our programs are higher than LCU and are omitted below.

\paragraph{Results}
In \Cref{tbl:algo-benchmark-comparison}, we present the runtime cost for each benchmark and polynomial implementation, as calculated by the \LangName{} compiler immediately before circuit generation.
These results confirm that QSVT is favorable in cost for these benchmarks by several orders of magnitude.
They also indicate that subnormalization is the main bottleneck --- Horner's method achieves optimal query count, but its repeated use of arithmetic accumulates a large subnormalization.

\LangName{} automatically produces the efficient circuits intended by the authors of these algorithms and empirically confirms the theoretical prediction of efficiency of QSVT over LCU\@. To our knowledge, our system is among the first to achieve this goal for executable quantum programs.

\begin{table}
\centering
\caption{Runtime cost comparison for quantum algorithms in \LangName{}. Each benchmark is a polynomial approximation of a matrix function with given parity, truncated to degree $d$. Column ``LCU'' reports a baseline implementation of polynomials using linear combination of unitaries (\Cref{def:lcu}), ``Horner'' reports the use of Horner's method (\Cref{sec:polynomials}), and ``QSVT'' reports the use of the quantum singular value transform.}
\label{tbl:algo-benchmark-comparison}
\resizebox{\textwidth}{!}{\begin{tabular}{ l c c r r r }
  \toprule
  & & & \multicolumn{3}{c}{{\# Queries $\times$ Subnormalization = Cost}} \\
  \cmidrule{4-6}
  {Quantum Algorithm} & Parity & $d$ & \multicolumn{1}{c}{{LCU}} & \multicolumn{1}{c}{{Horner}} & \multicolumn{1}{c}{{QSVT}} \\
  \midrule
  Matrix inversion & odd & 7 & $49 \times \num{5.15e+04} = \num{2523395.9}$ & $\textbf{13} \times \num{5.15e+04} = \num{669472.4}$ & $\textbf{13} \times \textbf{5.4} = \textbf{69.8}$ \\
  Hamiltonian simulation & mixed & 10 & $120 \times \num{1066} = \num{127935.9}$ & $\textbf{15} \times \num{1066} = \num{15992.0}$ & $29 \times \textbf{2.0} = \textbf{58.0}$ \\
  Spectral thresholding & odd & 10 & $99 \times \num{5.27e+05} = \num{52174866.5}$ & $\textbf{19} \times \num{5.27e+05} = \num{10013358.2}$ & $\textbf{19} \times \textbf{2.7} = \textbf{51.6}$ \\
  \bottomrule
\end{tabular}}
\end{table}

\subsection{RQ3: Comparison to Existing Circuit Optimizers} \label{sec:eval-comparison}

\begin{enumerate}
  \item [\emph{RQ3.}] By how much do existing circuit optimizers reduce the costs of these programs?
\end{enumerate}

\paragraph{Benchmarks}
We used \LangName{} to compile each unoptimized program above to a quantum circuit in the gate set $\{H, X, \textrm{CNOT}, R_z(\theta)\}$ of \citet{nam2018}.
For sake of benchmarking, we instantiated each black-box oracle as a random but consistent single-qubit rotation gate.
We invoked on each circuit each of these quantum circuit optimizers: Quartz \citep{xu2022}, wisq \citep{xu2025}, Qiskit \citep{javadiabhari2024}, Feynman \citep{amy2014}, VOQC \citep{hietala2021}, Pytket \citep{sivarajah2020} (peephole and ZX), and QuiZX \citep{quizx}. Each optimizer was given input in the Nam gate set, the output of each was confirmed to remain in this gate set, and those that support user-selectable gate sets were explicitly invoked with this gate set.

\paragraph{Cost Metric}
We counted the number of qubits and \emph{non-Clifford} gates in the circuits output by each optimizer and the circuits generated by \LangName{} after the optimizations of \Cref{sec:optimizations}. These counts ignore subnormalization, pretending for simplicity that the circuit only runs once.

A \emph{non-Clifford} gate is not generated by products or tensor products of $\{H, \textrm{CNOT}, R_z(\pi / 2)\}$, and incurs significant overhead under predominant quantum error-correcting codes \citep{fowler2012}.
In the Nam gate set, the non-Clifford gates are $R_z(\theta)$ with $\theta$ not an integral multiple of $\pi/2$.

We chose the Nam gate set because its support for continuous rotations makes it the closest fit for QSVT circuits among standard gate sets supported by existing optimizers.
Using the Clifford+$T$ gate set is an alternative but would be subject to confounding effects.
Specifically, the $T$-count of an $R_z(\theta)$ gate is dictated by the specific value of $\theta$ as well as the algorithm and tolerance used for the  discretization of $\theta$ in unitary synthesis --- which are not under study in this work.

\begin{figure}
\tikzsetfigurename{optimization-comparison}
\centering
\begin{tikzpicture}
\begin{axis}[
    ybar,
    bar width=0.08,
    width=13.5cm,
    height=5.7cm,
    ymin=0,
    ymax=110,
    ytick={0,25,50,75,100},
    yticklabels={0\%,25\%,50\%,75\%,100\%},
    ylabel={Non-Clifford Gates Remaining},
    xlabel={},
    xmin=-0.3,
    xmax=12.5,
    xtick={1.1,3.1,5.1,7.1,9.1,11.1},
    xticklabels={simulation-example,regression-example,penalized-coupler,laplacian-filter,matrix-inversion,hamiltonian-simulation},
    x tick label style={rotate=20,anchor=east,font=\ttfamily\scriptsize},
    label style={font=\small},
    axis line style={draw=gray},
    xtick pos=left,
    legend cell align={left},
    legend style={
        at={(0.5,1.02)},
        anchor=south,
        legend columns=5,
        column sep=0.25em,
        /tikz/every even column/.append style={column sep=1em, text width=4.75em},
        draw=none,
        font=\footnotesize,
    },
    bar shift=0pt,
    cycle list={
        {blue!50!white,fill=blue!50!white},
        {teal!50!white,fill=teal!50!white},
        {green!50!white,fill=green!50!white},
        {lime,fill=lime},
        {yellow,fill=yellow},
        {orange!50!white,fill=orange!50!white},
        {red!50!white,fill=red!50!white},
        {magenta!50!white,fill=magenta!50!white},
        {violet,fill=violet},
    },
]

\addplot coordinates {    
    (0.7,68/100*100)      
    (2.7,69/71*100)       
    (4.7,148/160*100)     
    (6.7,798/806*100)     
    (8.7,14120/14120*100) 
};

\addplot coordinates {     
    (0.8,18/100*100)       
    (2.8,30/71*100)        
    (4.8,65/160*100)       
    (6.8,360/806*100)      
    (8.8,14023/14120*100)  
    (10.8,66420/66420*100) 
};

\addplot coordinates {     
    (0.9,100/100*100)      
    (2.9,58/71*100)        
    (4.9,137/160*100)      
    (6.9,688/806*100)      
    (8.9,12239/14120*100)  
    (10.9,57090/66420*100) 
};

\addplot coordinates {     
    (1.0,58/100*100)       
    (3.0,53/71*100)        
    (5.0,94/160*100)       
    (7.0,514/806*100)      
    (9.0,10000/14120*100)  
    (11.0,47304/66420*100) 
};

\addplot coordinates {     
    (1.1,14/100*100)       
    (3.1,54/71*100)        
    (5.1,55/160*100)       
    (7.1,314/806*100)      
    (9.1,5266/14120*100)   
    (11.1,21070/66420*100) 
};

\addplot coordinates {     
    (1.2,68/100*100)       
    (3.2,34/71*100)        
    (5.2,137/160*100)      
    (7.2,584/806*100)      
    (9.2,10501/14120*100)  
};

\addplot coordinates {     
    (1.3,32/100*100)       
    (3.3,52/71*100)        
    (5.3,85/160*100)       
    (7.3,394/806*100)      
    (9.3,6038/14120*100)   
    (11.3,30387/66420*100) 
};

\addplot coordinates {     
    (1.4,2/100*100)        
    (3.4,52/71*100)        
    (5.4,53/160*100)       
    (7.4,218/806*100)      
    (9.4,3114/14120*100)   
};

\addplot coordinates {     
    (1.5,2/100*100)        
    (3.5,27/71*100)        
    (5.5,34/160*100)       
    (7.5,18/806*100)       
    (9.5,53/14120*100)     
    (11.5,631/66420*100)   
};

\legend{Quartz,wisq,Qiskit,Feynman,VOQC,Pytket,Pytket ZX,QuiZX,\LangName{} (Ours)}
\addlegendimage{ybar, fill=white, pattern=north east lines, pattern color=gray, draw=none, draw opacity=0}
\addlegendentry{No result}

\addplot[ybar, pattern=north east lines, pattern color=blue!50!white, draw=none] coordinates {(10.7,100)};
\addplot[ybar, pattern=north east lines, pattern color=orange!50!white, draw=none] coordinates {(11.2,100)};
\addplot[ybar, pattern=north east lines, pattern color=magenta!50!white, draw=none] coordinates {(11.4,100)};

\node[anchor=north, color=darkgray] at (axis cs:0.18,110) {\scriptsize 100\% =};
\node[anchor=north, color=darkgray] at (axis cs:1.1,110) {\scriptsize 100 gates};
\node[anchor=north, color=darkgray] at (axis cs:3.1,110) {\scriptsize 71 gates};
\node[anchor=north, color=darkgray] at (axis cs:5.1,110) {\scriptsize 160 gates};
\node[anchor=north, color=darkgray] at (axis cs:7.1,110) {\scriptsize 806 gates};
\node[anchor=north, color=darkgray] at (axis cs:9.1,110) {\scriptsize 14120 gates};
\node[anchor=north, color=darkgray] at (axis cs:11.1,110) {\scriptsize 66420 gates};

\node[anchor=west, font=\small, color=violet] at (axis cs:1.46,5) {$2.0\%$};
\node[anchor=west, font=\small, color=violet] at (axis cs:3.46,41.0) {$38.0\%$};
\node[anchor=west, font=\small, color=violet] at (axis cs:5.46,24.2) {$21.2\%$};
\node[anchor=west, font=\small, color=violet] at (axis cs:7.46,5.2) {$2.2\%$};
\node[anchor=west, font=\small, color=violet] at (axis cs:9.46,4.0) {$0.4\%$};
\node[anchor=west, font=\small, color=violet] at (axis cs:11.46,5.0) {$1.0\%$};

\end{axis}
\end{tikzpicture}
\setlength{\abovecaptionskip}{1pt}%
\caption{Gate counts of circuits across benchmarks and optimizers. Each bar shows a normalized fraction of gates of the unoptimized circuit (lower is better). Slashed bars denote optimizers that crashed, used ${>}32$ GB of memory, or timed out after one hour. A benchmark is shown if ${>}3$ existing optimizers ran to completion.}
\label{fig:optimization-comparison}
\end{figure}

\begin{table}
\centering
\vspace*{1ex}%
\caption{Qubit counts of circuits across benchmarks and optimizers, relative to the unoptimized baseline program. An entry with -- indicates that the optimizer did not run to completion to produce a circuit.}
\label{tbl:qubit-reduction}
\resizebox{\textwidth}{!}{%
\begin{tabular}{ l r rrrrrrrrr }
  \toprule
  & & \multicolumn{9}{c}{{\# Qubits (\% Reduction from Unoptimized Baseline)}} \\
  \cmidrule(lr){2-11}
  {Benchmark} & {Base} & \multicolumn{1}{c}{Quartz} & \multicolumn{1}{c}{wisq} & \multicolumn{1}{c}{Qiskit} & \multicolumn{1}{c}{Feynman} & \multicolumn{1}{c}{VOQC} & \multicolumn{1}{c}{Pytket} & \multicolumn{1}{c}{Pytket ZX} & \multicolumn{1}{c}{QuiZX} & \multicolumn{1}{c}{Cobble} \\
  \midrule
  \texttt{simulation-example} & 5 & 5 (0\%) & 5 (0\%) & 5 (0\%) & 5 (0\%) & 5 (0\%) & 5 (0\%) & 5 (0\%) & 5 (0\%) & 3 (40.0\%) \\
  \texttt{regression-example} & 6 & 6 (0\%) & 6 (0\%) & 6 (0\%) & 6 (0\%) & 6 (0\%) & 6 (0\%) & 6 (0\%) & 6 (0\%) & 3 (50.0\%) \\
  \texttt{penalized-coupler} & 7 & 7 (0\%) & 7 (0\%) & 7 (0\%) & 7 (0\%) & 7 (0\%) & 7 (0\%) & 7 (0\%) & 7 (0\%) & 5 (28.6\%) \\
  \texttt{laplacian-filter} & 11 & 11 (0\%) & 11 (0\%) & 11 (0\%) & 11 (0\%) & 11 (0\%) & 11 (0\%) & 11 (0\%) & 11 (0\%) & 6 (45.5\%) \\
  \texttt{matrix-inversion} & 12 & 12 (0\%) & 12 (0\%) & 12 (0\%) & 12 (0\%) & 12 (0\%) & 12 (0\%) & 12 (0\%) & 12 (0\%) & 2 (83.3\%) \\
  \makebox[10em][l]{\texttt{hamiltonian-simulation}} & 16 & -- & 16 (0\%) & 16 (0\%) & 16 (0\%) & 16 (0\%) & -- & 16 (0\%) & -- & 6 (62.5\%) \\
  \midrule
  \Cref{fig:optimization-comparison} Average &  & (0\%) & (0\%) & (0\%) & (0\%) & (0\%) & (0\%) & (0\%) & (0\%) & (51.6\%) \\
  \midrule
  \texttt{ols-ridge} & 22 & -- & 22 (0\%) & 22 (0\%) & -- & -- & -- & -- & -- & 2 (90.9\%) \\
  \makebox[0pt][l]{\texttt{spectral-thresholding}} & 16 & -- & 16 (0\%) & 16 (0\%) & 16 (0\%) & -- & -- & -- & -- & 2 (87.5\%) \\
  \midrule
  Overall Average &  & (0\%) & (0\%) & (0\%) & (0\%) & (0\%) & (0\%) & (0\%) & (0\%) & (61.0\%) \\
  \bottomrule
\end{tabular}
}
\end{table}

\paragraph{Results}
In \Cref{fig:optimization-comparison}, we present the comparison of non-Clifford gate counts achieved by different optimizers and \LangName{} across various benchmarks. \LangName{}'s optimizations, which operate on high-level program structure rather than circuits, tie or exceed the performance of all of the evaluated circuit optimizers, and its relative performance tends to improve for larger programs.

These results show how specializing compilers to program structure can be useful for complex quantum applications. Circuit optimizers not aware of the algebraic structure of block encodings cannot remove high-level redundancies as easily. They are also restricted by their typical design, which obligates them to strictly (or very closely) preserve the semantics of the input circuit. By contrast, sum and polynomial fusion reduce subnormalization --- preserving the block encoding semantics (\Cref{sec:denotational-semantics}) but not the overly conservative circuit semantics of the program.

With that said, a developer can benefit from invoking sum and polynomial fusion in \LangName{} followed by an existing circuit optimizer to obtain orthogonal and compounding improvements.
For \texttt{hamiltonian-simulation}, the additional gate reduction from running a circuit optimizer on the output of \LangName{} ranges from 12.8\% (Qiskit) to 86.0\% (QuiZX).
The other programs in \Cref{fig:optimization-comparison} become too small after sum and polynomial fusion to permit meaningful comparison.

Our results show that existing circuit optimizers do not change the qubit usage for the programs in \Cref{fig:optimization-comparison}, whereas \LangName{} reduces it by an average of 52\% (see \Cref{tbl:qubit-reduction}; 29\% for \texttt{penalized-coupler} up to 83\% for \texttt{matrix-inversion}). Because the specific circuits analyzed here are mainly bottlenecked by gates and not qubits, this reduction is not a primary emphasis of our study.

\subsection{RQ4: Scalability in Compile Time}

\begin{enumerate}
  \item [\emph{RQ4.}] How scalable is the \LangName{} compiler in compile time with varying problem size?
\end{enumerate}

\paragraph{Benchmarks}
\begin{wrapfigure}[14]{r}{.45\textwidth}
\vspace*{-3.75ex}%
\tikzsetfigurename{fig:compile-time}%
\centering
\begin{tikzpicture}%
\begin{axis}[
    xlabel={Total Gates of Unoptimized Program},
    ylabel={Compile Time (seconds)},
    legend pos=north west,
    grid=major,
    draw=gray,
    width=6cm,
    height=5cm,
    xmin=0,
    xmax=2593612,
    ymin=0,
    ymax=0.2,
    scaled x ticks=false,
    scaled y ticks=false,
    xtick={0,5e5,1e6,1.5e6,2e6,2.5e6},
    xticklabels={0,$5\textrm{e}5$,$1\textrm{e}6$,$1.5\textrm{e}6$,$2\textrm{e}6$,$2.5\textrm{e}6$},
    y tick label style={/pgf/number format/fixed},
    legend cell align={left},
    legend style={font=\scriptsize},
    label style={font=\small},
]

\addplot[
    fill=green!30!white,
    draw=green!50!black,
    mark=none,
    area legend,
    stack plots=y
] coordinates {
    (190, 0.0006)
    (412, 0.0007)
    (1946, 0.0009)
    (3814, 0.0010)
    (6292, 0.0011)
    (8779, 0.0012)
    (25871, 0.0014)
    (41053, 0.0015)
    (58440, 0.0017)
    (77229, 0.0018)
    (98268, 0.0020)
    (120887, 0.0021)
    (145707, 0.0023)
    (170559, 0.0025)
    (407101, 0.0027)
    (529459, 0.0030)
    (660093, 0.0031)
    (797400, 0.0033)
    (943054, 0.0035)
    (1095813, 0.0037)
    (1256844, 0.0039)
    (1424528, 0.0041)
    (1600745, 0.0043)
    (1783958, 0.0046)
    (1975449, 0.0048)
    (2173571, 0.0051)
    (2380064, 0.0053)
    (2593630, 0.0056)
} \closedcycle;
\addlegendentry{Non-solver time}

\addplot[
    fill=blue!30!white,
    draw=blue!50!black,
    mark=none,
    area legend,
    stack plots=y
] coordinates {
    (190, 0.0074)
    (412, 0.0070)
    (1946, 0.0115)
    (3814, 0.0115)
    (6292, 0.0175)
    (8779, 0.0173)
    (25871, 0.0241)
    (41053, 0.0245)
    (58440, 0.0326)
    (77229, 0.0328)
    (98268, 0.0423)
    (120887, 0.0425)
    (145707, 0.0537)
    (170559, 0.0536)
    (407101, 0.0662)
    (529459, 0.0706)
    (660093, 0.0807)
    (797400, 0.0804)
    (943054, 0.0954)
    (1095813, 0.0960)
    (1256844, 0.1125)
    (1424528, 0.1133)
    (1600745, 0.1303)
    (1783958, 0.1301)
    (1975449, 0.1496)
    (2173571, 0.1500)
    (2380064, 0.1704)
    (2593630, 0.1707)
} \closedcycle;
\addlegendentry{External solver time}

\end{axis}
\end{tikzpicture}
\setlength{\abovecaptionskip}{1ex}
\caption{Compile time taken by \LangName{} using the pyQSP external solver. The standard error of the mean is less than 0.001 seconds throughout.}
\label{fig:compile-time}
\end{wrapfigure}

For this research question, we implemented a family of programs that block-encode the Chebyshev polynomials $T_n(A)$, which scale uniformly for testing, for $2 \le n \le 30$. We executed the \LangName{} compiler on each program with all optimizations enabled and with pyQSP as the external solver.
We measured the time taken to produce a circuit as the average of 10 samples. All timing results were collected on a 2.4 GHz Intel Core i9 processor.

\paragraph{Results}
In \Cref{fig:compile-time}, we present how the compile time used by \LangName{} scales with the gate count of the unoptimized program. As shown in the graph, \LangName{} can process and optimize a high-level program whose unoptimized circuit would contain millions of logic gates, in a fraction of a second.
Much of the time is spent in a single call to the external numerical solver for QSP phase angles from polynomial coefficients, a computationally intensive step \citep{dong2021} reported as a separate component of the total.
\LangName{}'s compilation speed could be directly improved by reducing overhead from the QSP solver or from the Python interpreter.

\section{Limitations and Next Steps}

In this section, we discuss current limitations of \LangName{} and opportunities for future work.

\paragraph{Basic Matrix Encodings}
\LangName{} currently asks a user to provide block encodings for basic matrices, which range from trivial unitary gates to potentially complex sub-circuits, before composing them via arithmetic operators. Basic matrices are where compositionality appears to end and case-by-case reasoning must begin.
An important next step is to design and provide abstractions for known matrix structures, such as banded, circulant, and Toeplitz matrices, from the literature \citep{sunderhauf2024, camps2022, camps2024}. A long-term goal is to develop domain-specific languages for a broader range of matrix structures relevant to applications.

\paragraph{Commutativity}
\LangName{} currently asks a user to specify whether a product of matrices commutes. Automatic checking of this condition is a hard problem today in circuit optimization, since commutativity is not an inductively defined property on matrices and numerical evaluation is not scalable.
A first step toward more automation is to hard-code important cases, such as the product of Pauli matrices. A longer term solution could be to analyze the algebraic structure of matrices using tools from representation theory or the ZX-calculus \citep{coecke2011}.

\paragraph{Rewrite Rules}
\LangName{} currently uses a set of manually implemented rewrites that invite exploration of completeness or automation.
This understanding must be grounded in new benchmarks, which are still emerging. Initial steps, however, include systematically applying rules via equality saturation \citep{tate2011} and verifying rules via proof assistants \citep{liu2022}. In the longer term, program synthesis \citep{xu2023} could help discover new rules automatically.

\paragraph{Approximation}
Like many other compilers, \LangName{} does not fully account for the approximation inherent in quantum algorithms, whose cost must be calculated and traded off to execute programs on resource-constrained hardware such as near-term noisy devices or emerging error-correcting codes.
A step toward better error and cost analysis is to formalize in the language the precision parameter $\epsilon$ of block encodings. In the longer term, methods such as those of \citet{hung2019} could help track error in QSP phase calculation, $\lambda$ state preparation, and $R_z$ synthesis.

\section{Related Work}

In this section, we survey the research that is most closely relevant to this work.

\paragraph{Quantum Programming Languages}
Researchers have developed many languages for expressing, analyzing, and verifying quantum algorithms.
Recent ideas in language design beyond circuits include automatic uncomputation \citep{bichsel2020}, pointers and memory \citep{tower}, classical effects \citep{voichick2023}, control flow in superposition \citep{qcm}, type systems for basis \citep{adams2024}, and call stacks for recursion \citep{zhang2025}.

Other work in the community has led to more effective compilers for quantum programs \citep{huang2024,sharma2025,tornow2025,peduri2022,peng2024} as well as more robust verification tools \citep{huang2025,xu2025b,chen2023}.

Our work introduces a new programming abstraction targeting a new domain --- mathematical operators for block encodings in quantum computational linear algebra. We hope these insights will lead to languages that enable developers to realize broader classes of algorithms.

\paragraph{Quantum Resource Estimation}

Parallel to the design of languages and compilers is the creation of quantum resource estimation frameworks such as Qualtran \citep{harrigan2024}, Bartiq \citep{psiquantum2024}, Azure QRE \citep{vandam2023}, and pyLIQTR \citep{obenland2025}. Given a complex algorithm that is hard to compile directly, these frameworks typically emphasize the ability to produce practical hardware cost estimates rather than explicit executable circuits.

Our work can help in developing these frameworks, which have emerging support for quantum linear algebra, into more full-fledged compilers.
We build on prior work on Qualtran \citep{harrigan2024} that observed a limited case of sum fusion by formalizing the general technique in a language with sound semantics and evaluating it across a broad set of benchmarks.

\paragraph{Quantum Singular Value Transformation}
Block encodings and the QSVT culminate a long line of research in quantum linear algebra. They generalize prior representations of matrices in algorithms for simulation \citep{low2019} and matrix inversion \citep{childs2017}, while unifying these tasks with search and phase estimation \citep{martyn2021}. Given that more general assumptions of unstructured matrix access can erase quantum speedup \citep{tang2019}, block encodings define the structure of matrices that are relevant and efficient to encode into a quantum computer.

Interest in the QSVT has led to new numerical solvers for QSP phase angles \citep{dong2021,berntson2025,alexis2026}, adaptation to mixed-parity polynomials \citep{motlagh2024,sunderhauf2023}, and even hardware realization \citep{kikuchi2023}.
Procedures to build the QSVT circuit (\Cref{fig:qsvt-circuit}) from an input list of phase angles have recently been added to the PennyLane \citep{bergholm2022} and Qmod \citep{vax2025} programming frameworks.

Our work provides a roadmap to deploy this mathematical toolkit: a high-level language and compiler that instantiate and test existing solvers and circuit constructions, a type system that checks required conditions for QSVT matrices and polynomials, and an empirical comparison of QSVT against other methods to evaluate where it gives the biggest gains in practice.

\paragraph{Classical Linear Algebra}
Ever since BLAS \citep{blackford2002} and LAPACK \citep{anderson1990}, linear algebra has been central in programming, compilers, and high-performance computing.
Researchers have developed numerous methods \citep{whaley1998,tomov2010,anzt2022, vanzee2009,kjolstad2017,puschel2005,frigo2005} to automatically tailor linear algebra to workloads and architectures. Machine learning has further accelerated interest in practical and scalable compilers \citep{paszke2019,abadi2016,frostig2018,sabne2020,chen2018}.
Of these tools, perhaps the closest analogue to \LangName{} is Eigen \citep{gunnebaud2010}, a user-facing library and optimizer for matrix expressions.

Our work takes a step toward similar goals in expressing and optimizing quantum algorithms. It also suggests that our assumptions about how to optimize programs, such as computation graphs that rely on unrestricted sharing of subexpressions, should be revisited in the quantum world.

\section{Conclusion}

Modern quantum algorithms for linear algebra constitute many of our hopes to reap the reward of practical computational advantage from our investment into hardware quantum devices.
Like an intricate puzzle, they present new challenges -- building and optimizing complex programs -- and new insights that change the way we think about and interact with computation.

In this work, we present high-level programming abstractions for quantum linear algebra, which we hope will help unlock this rare domain of potentially exponential speedup.
Through the sum and polynomial fusion optimizations, we also hope to catalyze the discovery of powerful structural rewrites analogous to loop-invariant code motion and strength reduction from classical compilers.
And by building systems that reduce need for expert-level reasoning about qubits and logic gates, we hope to eventually enable a broad range of programmers to use a quantum computer.

\begin{acks}
We thank Isaac Chuang and Patrick Rall for providing an introduction to the theory and applications of block encodings. We thank Matthew Harrigan for providing an environment to explore related ideas and Anurudh Peduri for invaluable technical discussions. We thank Swamit Tannu and Aws Albarghouthi for feedback on drafts of this work. We thank the Center for High-Throughput Computing at the University of Wisconsin--Madison for providing computational resources. Support for this research was provided by the Office of the Vice Chancellor for Research at the University of Wisconsin--Madison with funding from the Wisconsin Alumni Research Foundation.
\end{acks}

\section*{Data Availability Statement}

The artifact for this paper, including source code, benchmark programs, and evaluation package, is available on Zenodo~\citep{artifact}.

\bibliography{biblio}

\end{document}